\documentclass[12pt]{article}

\usepackage[margin=10ex]{geometry}
\usepackage{hyperref}

\usepackage{amsfonts}
\usepackage{amsmath,amssymb}
\usepackage{graphicx,amscd,xypic}
\usepackage{bbding}
\usepackage{indentfirst}
\usepackage{mathrsfs, dsfont, bbm}
\usepackage[T1]{fontenc}
\usepackage{color}
\usepackage{enumerate}

\newtheorem{thm}{Theorem}[section]

\newtheorem{prop}{Proposition}[section]

\newtheorem{lemma}[thm]{Lemma}

\newtheorem{assump}{Assumption}[section]
\newtheorem{case}{Case}[section]

\newenvironment{proof}{\hspace{0ex}\textsc{Proof}:\hspace{1ex}}{\hfill$\Box$\\[2ex] }

\DeclareMathOperator{\BP}{\mathbb{P}}
\DeclareMathOperator{\BF}{\mathbf{F}}

\allowdisplaybreaks

\begin{document}

\title{Optimal Insurance with Rank-Dependent Utility and Increasing Indemnities}

\author{Zuo Quan Xu\thanks{Department of Applied Mathematics, Hong Kong Polytechnic University, Kowloon, Hong Kong. Email: \href{mailto:maxu@polyu.edu.hk}{\nolinkurl{maxu@polyu.edu.hk}}. This author acknowledges financial supports from Hong Kong Early Career Scheme (No. 533112), Hong Kong General Research Fund (No. 529711) and NNSF of China (No. 11471276). } \and
Xun Yu Zhou\thanks{Mathematical Institute and Nomura Centre for Mathematical Finance, and Oxford--Man Institute of Quantitative Finance, The University of Oxford,  Oxford OX2 6GG, UK. This author acknowledges supports from a start-up fund of the University of Oxford, and research grants from the Nomura Centre for Mathematical Finance and the Oxford--Man Institute of Quantitative Finance.} \and Shengchao Zhuang\thanks{Department of Statistics and Actuarial Science, University of Waterloo, Waterloo, Ontario, N2L 3G1, Canada.}}

\maketitle

\begin{abstract}
Bernard et al. (2015) study an optimal insurance design problem where an individual's preference is of the rank-dependent utility (RDU) type, and show that in general an optimal contract covers both large and small losses. However, their contracts suffer from a problem of moral hazard for paying more compensation for a smaller loss. This paper addresses this setback by exogenously imposing the constraint that both the indemnity function and the insured's retention function be increasing with respect to the loss. We characterize the optimal solutions via calculus of variations, and then apply the result to obtain explicitly expressed contracts for problems with Yaari's dual criterion and general RDU.  Finally, we use a numerical example to compare the results between ours and that of Bernard et al. (2015).

\bigskip

\end{abstract}

%\newpage
%%%%%%%%%%%%%%%%%%%%%%%%%%%%
\section{Introduction}
%%%%%%%%%%%%%%%%%%%%%%%%%%%%
\noindent
Optimal insurance contract design is an important problem, manifested not only in theory but also in insurance and financial practices. The  problem is to determine the optimal amount of compensation as a function of the loss -- called indemnity -- so as to maximize the insured's satisfaction, subject to the participation constraint of the insurer.
\par
In the insurance literature, most of the work assume that the insurer is risk neutral while the insured is a risk-averse expected utility (EU) maximizer; see e.g. Arrow (1963), Raviv (1979), and Gollier and Schlesinger (1996). In this case, the optimal contract is in general a deductible one that covers part of the loss in excess of a deductible level.
\par
However, the EU theory has received many criticisms, for it fails to explain numerous experimental observations and theoretical puzzles. In the context of insurance contracting, the classical EU-based models cannot explain some behaviors in insurance demand such as that for small losses (e.g. demand for
warranties); see a detailed discussion in Bernard et al. (2015).
%For example, it fails to explain the famous Allais Paradox nor the reason why a same person may buy both lottery and insurance. Other paradoxes/puzzles that EU theory contradict include common ratio effect (Allais, 1953), Friedman and Savage puzzle (Friedman and Savage, 1948), Ellesberg paradox (Ellesberg, 1961), and the equity premium puzzle (Mehra and Prescott, 1985). Sydnor (2010) investigated how people choose the deductible decisions between \$100, \$250, \$500, and \$1,000. The major finding is that the households who pay an average premium of \$715 per year mostly choose the deductible of \$500. However, the policy with a \$1000 deductible that just requires the premium of \$615 every year was rejected by all the households. Assume that the claim rate is about 5 percent, these households were willing to to pay \$100 to protect against a 5 percent possibility of paying an additional \$500. As explained by Barberis (2012), this choice can only be explained by unreasonably high levels of risk aversion from the perspective of expected utility framework.
\par
In order to overcome this drawback of the EU theory, different measures of evaluating uncertain outcomes have been put forward to depict human behaviors. A notable one is the rank-dependent utility (RDU) proposed by Quiggin (1982), which consists of a concave utility function and an inverse-$S$ shaped probability weighting (or distortion) function.\footnote{The RDU preference reduces to Yaari's dual criterion (Yaari 1987) when the utility function is the identity one.} Through the probability weighting, the RDU theory captures the common observation that people tend to exaggerate small probabilities of extremely good and bad outcomes. With the development of advanced mathematical tools, the RDU preference has been applied to many areas of finance, including portfolio choice and option pricing. On the other hand, Barseghyan et al. (2013) use data on households' insurance deductible decisions in auto and home insurance to demonstrate the relevance and importance of the probability weighting and suggest the possibility of generalizing their conclusions to other insurance choices.
\par
There have been also studies in the area of insurance contract design within the RDU framework; see for example Chateauneuf, Dana and Tallon (2000), Dana and Scarsini (2007), and Carlier and Dana (2008). However, all these papers assume that the probability weighting function is convex. Bernard et al. (2015) are probably the first to study RDU-based insurance contracting  with  {\it inverse-$S$ shaped} weighting functions, using the quantile formulation originally developed for portfolio choice (Jin and Zhou 2008, He and Zhou 2011).  They derive optimal contracts that not only insure large losses above a deductible level but also cover small ones. However, their contracts suffer from a severe problem of moral hazard, since they are not increasing with respect to the losses.\footnote{Throughout this paper, by an ``increasing'' function we mean a ``non-decreasing'' function, namely $f$ is increasing if $f(x)\geqslant f(y)$ whenever $x>y$. We say $f$ is ``strictly increasing'' if  $f(x)> f(y)$ whenever $x>y$. Similar conventions are used for ``decreasing'' and ``strictly decreasing'' functions.} As a consequence, insureds may be motivated to hide their true losses in order to obtain additional compensations;  see a discussion on pp. 175--176 of Bernard et al. (2015).
\par
This paper aims to address this setback. We consider the same insurance model as in Bernard et al. (2015), but adding an explicit constraint that both the indemnity function and the insured's retention function (i.e. the part of the losses to be born by the insured) must be globally increasing with respect to the losses. This constraint will rule out completely the aforementioned behaviour of moral hazard; yet mathematically it gives rise to
substantial difficulty. The approach used in Bernard et al. (2015) no longer works. We develop a general approach to overcome this difficulty. Specifically, we first derive the necessary and sufficient conditions for optimal solutions via calculus of variations. Then we deduce  explicitly expressed optimal contracts by a fine analysis on these conditions. An interesting finding is that, for a good and reasonable range of parameters specifications, there are only two types of optima contracts, one being the classical deductible one and the other a ``three-fold" one covering both small and large losses.
\par
The remainder of the paper is organized as follows. Section 2 presents the optimal insurance model under the RDU framework including its quantile formulation. Section 3 applies the calculus of variations to derive a general necessary and sufficient condition for optimal solutions. We then derive optimal contracts for Yaari's criterion and the general RDU in Sections 4 and 5, respectively.  Section 6 provides a numerical example to illustrate our results. Finally, we conclude with Section 7. Proofs of some lemmas are placed in an Appendix.

%%%%%%%%%%%%%%%%%%%%%%%%%%%%
\section{The Model}
%%%%%%%%%%%%%%%%%%%%%%%%%%%%
\noindent
In this section, we present the optimal insurance contracting model in which the insured has the RDU type of preferences, followed by
its quantile formulation that will facilitate deriving the solutions.

\subsection{Problem formulation}
\noindent
We follow Bernard et al. (2015) for the problem formulation except for one critical difference, which we will highlight. Let $(\Omega, \BF, \BP )$ be a probability space. An insured, endowed with an initial wealth $W_0$, faces a non-negative random loss $X$ supported in $[0,M]$, where $M$ is a given positive scalar. He chooses an insurance contract to protect himself from the loss, by paying a premium $\pi$ to the insurer in return for a compensation (or {\it indemnity}) in the case of a loss. This compensation is to be determined as a function of the loss $X$, denoted by $I(\cdot)$ throughout this paper.
The {\it retention} function $R(X):=X-I(X)$ is thereby the part of the loss to be borne by the insured.

For a given $X$, the insured aims to choose an insurance contract that provides the best tradeoff between the premium and compensation based on his
risk preference. In this paper, we consider the case when insured's preference is of the RDU type. This RDU preference consists of two  components:
a $utility$ function $u:\mathbb{R}^+\mapsto \mathbb{R}^+$ and a probability $weighting$ function $T:[0,1]\mapsto [0,1]$.
Let us denote by $V^{rdu}(W)$ the RDU of the final (random) wealth $W$ of an insured, which is a Choquet integral of $u(W)$ with respect to the capacity $T\circ \BP $, i.e.,
\begin{eqnarray}\label{rdeu}
\nonumber
V^{rdu}(W)= && \int u(W)d(T\circ \BP ):= \int_{\mathbb{R}^+} u(x)d[-T(1-F_W(x))],
\end{eqnarray}
where $F_W(\cdot)$ is the cumulative distribution function (CDF) of $W$. Assuming that $T$ is differentiable,
we can rewrite
\[
 V^{rdu}(W)=\int_{\mathbb{R}^+} u(x)T'(1-F_W(x)) dF_W(x).
\]
If $T$ is inverse-$S$ shaped, that is, it is first concave and then convex; see Figure 1, then the above expression shows that the role $T$ plays is to overweigh both tails of $W$
when evaluating the mean of $u(W)$.
On the other hand, if the insurer is risk-neutral and the cost of offering the compensation is proportional to the expectation of the indemnity, then
 the premium to be charged for an insurance contract should satisfy the participation constraint
\[ \pi\geqslant (1+\rho)E[I(X)], \]
where the constant  $\rho$ is the {\it safety loading} of the insurer.
\par
It is natural to require an indemnity function to satisfy
\begin{align}\label{inceasing indeminity}
I(0)=0,\quad 0\leqslant I(x) \leqslant  x, \quad\forall \ 0\leqslant x\leqslant M,
\end{align}
a constraint that
has been imposed in most insurance contracting literature. If the insured's preference is dictated by the classical EU theory, then
the optimal contract is typically a deductible contract which automatically renders the indemnity function increasing; see e.g.
Arrow (1971) and Raviv (1979). However, for the RDU preference the resulting optimal
indemnity may not be an increasing function, as shown in Bernard et al. (2015). This may potentially cause moral hazard as pointed out earlier. Similarly, a non-monotone retention function may also lead to moral hazard. Consequently, to include the increasing constraint on the contract has been an outstanding open question.
 \par
In this paper, we require the indemnity function to satisfy  $I(0)=0$ and $0\leqslant I(x)-I(y)\leqslant x-y, \ \forall \  0\leqslant y<x\leqslant M$.
In other words, we constrain {\it both} indemnity and retention functions to be globally
increasing.
\par
We can now formulate our insurance contracting problem as
\begin{align}\label{orgi}
  \begin{array}{ll}
  \underset{{I(\cdot)}}{\text{max}} & \ \ \  V^{rdu}(W_0-\pi-X+I(X)) \\
  \text{s.t.} & \quad (1+\rho)E[I(X)]\leqslant  \pi,  \\
             & \quad  I(\cdot)\in \mathbb{I},
 \end{array}
\end{align}
where
\begin{align*}
  \mathbb{I}:=\{I(\cdot): I(0)=0, \ 0\leqslant I(x)-I(y)\leqslant x-y, \ \forall \  0\leqslant y\leqslant x\leqslant M\},
\end{align*}
and $W_0$ and $\pi$ are fixed scalars.
\par
For any random variable $Y \geqslant 0 $ a.s., define the quantile function of $Y$ as
\[ F^{-1}_Y(t):=\inf \{x \in \mathbb{R^+}: P(Y \leqslant x)\geqslant t \},\;\; t\in [0,1].
 \]
Note that any quantile function is nonnegative,  increasing and left-continuous (ILC).
 \par
We now introduce the following assumptions that will be used hereafter.
\begin{assump}\label{assump:21}
The random loss $X$ has a strictly increasing distribution function $F_X$. Moreover, $F^{-1}_X$ is absolutely continuous on $[0,1]$.
\end{assump}

\begin{assump}\label{assump:22}
 (Concave Utility) The  utility function $u:\mathbb{R}^+\mapsto \mathbb{R}^+$ is strictly increasing and continuously differentiable. Furthermore, $u'$ is decreasing.
\end{assump}
\begin{assump}\label{assump:23}
(Inverse-$S$ Shaped Weighting) The probability weighting function $ T $ is a continuous and strictly increasing mapping from [0,1] onto [0,1] and twice differentiable on $(0,1)$. Moreover, there exists $b\in(0,1)$ such that $T'(\cdot)$ is strictly decreasing on $(0,b)$ and strictly increasing on $(b,1)$. Furthermore, $T'(0+):=\lim_{z\downarrow 0} T'(z)>1 $
and $T'(1-):=\lim_{z\uparrow1} T'(z)=+\infty $.
\end{assump}
\par
The first part of Assumption \ref{assump:21}, crucial for the quantile formulation,  is standard; see e.g.  Raviv (1979). Note a significant difference from Bernard et al. (2015) is that here we allow $X$ to have atoms (which is usually the case in the insurance context).
For example, let  $X$ be distributed with $F_X(x)=\frac{1-\gamma e^{-\eta x}}{1-\gamma e^{-\eta M}}$ for $x\in [0,M]$, where $\gamma\in (0,1)$ and $\eta>0$. Then, $X$ satisfies Assumption \ref{assump:21}, and has an atom at $0$ with the probability $\mathbb{P}(X=0)=\frac{1-\gamma }{1-\gamma e^{-\eta M}}>0$. This assumption also ensures that $F^{-1}_X(F_{X}(x))\equiv x,\forall \ x \in [0,M]$, a fact that will be used often in the subsequent analysis.
Next, Assumption \ref{assump:22} is standard for a utility function. Finally, Assumption \ref{assump:23} is satisfied for many weighting functions proposed or used in the literature, e.g.  that proposed by Tversky and Kahneman (1992) (parameterized by $\theta$):
\begin{align}\label{TK-distortion}
T_\theta(x)=\frac{x^\theta}{(x^\theta+(1-x)^\theta)^{\frac{1}{\theta}}} .
\end{align}
Figure 1 displays this (inverse-$S$ shaped) weighting function (in blue) when $\theta=0.5$.
\begin{figure}[htb]\label{figure:distortion}
  \centering
      \includegraphics[width=4 in]{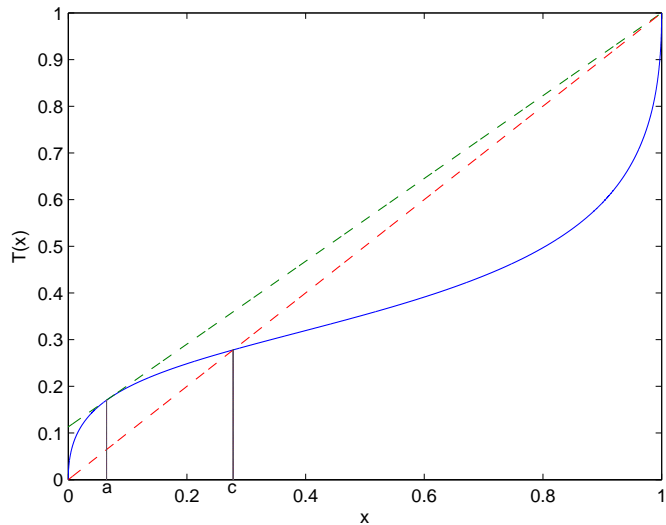}
      \caption{An inverse-$S$ shaped weighting function satisfying Assumption \ref{assump:23}. The marked points $a$ and $c$ will be explained later. }
\end{figure}
\par
In practice, most of the insurance contracts are not tailor-made for individual customers. Instead, an insurance company usually has contracts with different premiums to accommodate   customers with different needs. Each contract is designed with the best interest of a representative customer
in mind so as to stay marketable and competitive, while maintaining the desired profitability (the participation constraint). An insured can then choose one from the menu of contracts to cater for individual needs. Problem \eqref{orgi} is therefore motivated by the insurer's making of this menu.
\par
If the premium $\pi\geqslant (1+\rho)E[X]$, then $I^*(x)\equiv x$ (corresponding to a full coverage) is feasible and maximizes the objective function in  Problem \eqref{orgi} pointwisely; hence optimal. To rule out this trivial case, henceforth we restrict $0< \pi< (1+\rho)E[X]$. Moreover, we  assume
\begin{align}
W_0-(1+\rho)E[X]-M\geqslant 0,
\end{align}
to ensure that the policyholder will not go bankrupt. This is because $W_{0}-\pi-M>0$,  $\forall\; \pi\in (0,(1+\rho)E[X])$.
\par
It is more convenient to consider the retention function $R(x)=x-I(x)$ instead of $I(x)$ in
our study below. Letting $\Delta:=E[X]-\frac{\pi}{1+\rho} \in  (0,E[X]),\ W:=W_0-(1+\rho)E[X]>0,\ W_\Delta:=W+(1+\rho)\Delta\equiv W_0-\pi$, we can reformulate \eqref{orgi} in terms of $R(\cdot)$:
\begin{align}\label{orgi-2}
  \begin{array}{ll}
  \underset{{R(\cdot)}}{\text{max}} & \ \ \  V^{rdu}(W_\Delta-R(X)) \\
  \text{s.t.} & \quad E[R(X)]\geqslant \Delta,  \\
             & \quad  R(\cdot)\in \mathcal{R},
 \end{array}
\end{align}
where
\begin{align*}
  \mathcal{R}:=\{ R(\cdot): R(0)=0, \ 0\leqslant R(x)-R(y)\leqslant x-y, \ \forall \  0\leqslant y\leqslant x\leqslant M \}.
\end{align*}

\subsection{Quantile Formulation}
\noindent
The objective function in \eqref{orgi-2} is not concave in $R(X)$ (due to the nonlinear weighting function $T$), leading to a major difficulty in solving \eqref{orgi-2}.
However, under Assumption \ref{assump:23}, we have
\begin{eqnarray*}
\nonumber
V^{rdu}(W_\Delta-R(X))= && \int_{\mathbb{R}^+} u(x)d[-T(1-F_{W_\Delta-R(X)}(x))] \\
= && \int_0^1 u(F_{W_\Delta-R(X)}^{-1}(z))T'(1-z)dz= \int_0^1 u(W_\Delta-F_{R(X)}^{-1}(1-z))T'(1-z)dz \\
= &&  \int_0^1 u(W_\Delta-F_{R(X)}^{-1}(z))T'(z)dz,
\end{eqnarray*}
where the third equality is because $F_{W_\Delta-R(X)}^{-1}(z)=W_\Delta-F_{R(X)}^{-1}(1-z)$ except for a countable set of $z$. On the other hand, $E[R(X)]\geqslant \Delta$ is equivalent to $\int_0^1 F_{R(X)}^{-1}(z) dz\geqslant \Delta$.
\par
The above suggests that we may change the decision variable from the random variable $R(X)$ to its quantile function $F_{R(X)}^{-1}$, with which the objective function of \eqref{orgi-2} becomes concave and the first constraint is linear. It remains to rewrite the monotonicity constraint (represented by
the constraint set $\mathcal{R}$) also in terms of $F_{R(X)}^{-1}$. To this end, the next lemma plays an important role.

\begin{lemma}\label{lemma:quantile}
Under Assumption \ref{assump:21}, for any given $R(\cdot)\in \mathcal{R}$, we have \[R(x)=F^{-1}_{R(X)} (F_X (x)),\forall \ x \in [0,M].\]
\end{lemma}
\begin{proof}
%The result trivially holds when $x=0$. So we assume $0<x\leqslant M$.
First, by the monotonicity of $R(\cdot)$, we have  $\BP (R(X)\leqslant R(x)) \geqslant \BP (X\leqslant x) =F_X(x)$, so by the definition of $F^{-1}_{R(X)}(F_X(x))$, we conclude that $F^{-1}_{R(X)}(F_X(x))\leqslant R(x)$. It suffices to prove the reverse inequality.
Consider two cases.
\begin{itemize}
\item $R(x)=0$: In this case, we have $F^{-1}_{R(X)}(F_X(x))=0$ as quantile functions are nonnegative.
\item $R(x)>0$:  It  suffices to prove that $ \BP (R(X)\leqslant z)< F_X(x)$ for any $z < R(x)$. Take $z_1$ such that $z<z_1<R(x)$. By the continuity and monotonicity of $R(\cdot)$, there exists $y$ such that $y<x$ and $R(y)=z_1$. Then, $\BP (R(X)\leqslant z) \leqslant \BP (R(X) < z_1) = \BP (R(X) < R(y)) \leqslant  \BP (X \leqslant y)=F_X(y) < F_X(x)$, where we have used the fact that $F_X$ is strictly increasing under Assumption \ref{assump:21}.
\end{itemize}
The claim is thus proved.
\end{proof}

In view of the above results,  we can rewrite \eqref{orgi-2} as the following problem, in which the decision variable is $F_{R(X)}^{-1}(\cdot)$ (denoted by $G(\cdot)$ for simplicity):
\begin{equation}\label{orgi-3}
\begin{array}{rl}
  \max\limits_{ {G(\cdot)}}  & \quad  \int_0^1 u(W_\Delta-G(z))T'(z)dz, \\
  \mathrm{s.t.} & \quad  \int_0^1 G(z)dz \geqslant \Delta, \\
  & \quad  G(\cdot)\in {\mathbb{G}},
\end{array}
\end{equation}
 where $\mathbb{G}:=\{F^{-1}_{R(X)} (\cdot): R(\cdot)\in \mathcal{R}\}$.
 \par
In the absence of an explicit expression the constraint set $\mathbb{G}$ is hard  to deal with.
The following result addresses this issue. Note the major technical difficulty arises from the possible existence of the atoms of $X$.
\begin{lemma}
Under Assumption \ref{assump:21}, we have
\begin{equation}\label{Galternative}
{\mathbb{G}}=\{ G(\cdot): G(\cdot) \mbox{ is absolutely continuous, } \ G(0)=0, \ 0\leqslant G'(z)\leqslant (F^{-1}_X)'(z),   \mbox{ a.e. } z\in [0,1]  \}.
\end{equation}
\end{lemma}
\begin{proof}
We denote the right hand side of \eqref{Galternative} by $\mathbb{G}_{1}$.
For any $G(\cdot)\in\mathbb{G}$, there exists $R(\cdot)\in \mathcal{R}$ such that $G(\cdot)=F^{-1}_{R(X)}(\cdot)$.
For any $0\leqslant b<a \leqslant 1$,  define $\underline{a}=\inf\{x\in[0,M]: R(x)=G(a)\}$, $\overline{a}=\sup\{x\in[0,M]: R(x)=G(a)\}$, define $\underline{b}$ and $\overline{b}$ similarly. Let us show that $\underline{a} \leqslant F_{X}^{-1}(a)\leqslant \overline{a}$. In fact, by definition,
\begin{align*}
F^{-1}_X(a) &=\inf \{x \in \mathbb{R^+}: F_{X}(x)\geqslant a \} \geqslant \inf \{x \in \mathbb{R^+}: G(F_{X}(x))\geqslant G(a) \}\\
& =\inf \{x \in \mathbb{R^+}: R(x)\geqslant G(a) \}= \underline{a}.
\end{align*}
Suppose $F_{X}^{-1}(a)-\varepsilon>\overline{a}$ for some $\varepsilon>0$. Then by monotonicity,
\[G(a)=R(\overline{a})<R(F_{X}^{-1}(a)-\varepsilon)=G(F_X (F_{X}^{-1}(a)-\varepsilon))\leqslant G(a),\]
where we have used the fact that $F_X (F_{X}^{-1}(a)-\varepsilon)<a$  to get the last inequality. This leads to a contradiction; hence it must hold that  $F_{X}^{-1}(a)\leqslant \overline{a}$.
Similarly, we can prove $\underline{b} \leqslant F_{X}^{-1}(b)\leqslant \overline{b}$. Then we have
\[0\leqslant G(a)-G(b)=R(\underline{a})-R(\overline{b})\leqslant \underline{a}-\overline{b}\leqslant F^{-1}_X(a)-F^{-1}_X(b).\]
This inequality shows that  $G$ is absolutely continuous since $F_{X}^{-1}$ is an absolutely continuous function under Assumption \ref{assump:21}. Furthermore, it also implies \(0\leqslant G'(z)\leqslant (F^{-1}_X)'(z)),\) a.e. \(z\in [0,1].\)
So we have established  that $\mathbb{G}\subseteq\mathbb{G}_{1}$.
\par
To prove the reverse inclusion, take any $G(\cdot)\in \mathbb{G}_{1}$ and define $R(\cdot)=G(F_{X}(\cdot))$. It follows from  Assumption \ref{assump:21} that  $0\leqslant R(0)=G(F_{X}(0))-G(0)\leqslant F_{X}^{-1}(F_{X}(0))-F_{X}^{-1}(0)=0$ and $0\leqslant R(a)-R(b)=G(F_{X}(a))-G(F_{X}(b))\leqslant F_{X}^{-1}(F_{X}(a))-F_{X}^{-1}(F_{X}(b))=a-b$ $\forall 0\leqslant b<a \leqslant 1$. Hence $R(\cdot)\in\mathcal{R}$. It now suffices to show  $G(a)=F^{-1}_{R(X)}(a)$ for any  $0\leqslant a \leqslant 1$. If $G(a)=0$, then $G(a)\leqslant F^{-1}_{R(X)}(a)$ holds. Otherwise, for any $s<G(a)$, there exists $y$ such that $s<R(y)=G(F_X(y))<G(a)$ by the continuity of $R(\cdot)$. Then by the monotonicity of $R(\cdot)$ and $G(\cdot)$, we have \[\BP(R(X)\leqslant s)\leqslant \BP(R(X)<R(y)) \leqslant \BP(X\leqslant y)=F_X(y)<a,\] which means $G(a)\leqslant F^{-1}_{R(X)}(a)$. Using the same notation,  $\overline{a}$, as above, and
noting that $G(a)=R(\overline{a})=G(F_{X}(\overline{a}))$, we have $a \leqslant F_{X}(\overline{a})$ by the definition of $\overline{a}$ and the continuity of $R(\cdot)$. Moreover, it follows from \[\BP(R(X)\leqslant G(a))=\BP(R(X)\leqslant R(\overline{a})) = \BP(X\leqslant \overline{a})=F_{X}(\overline{a})\] that $F^{-1}_{R(X)}(F_{X}(\overline{a}))\leqslant G(a)$. Therefore,
\[G(a)\leqslant F^{-1}_{R(X)}(a) \leqslant F^{-1}_{R(X)}(F_{X}(\overline{a})) \leqslant G(a)\] holds by monotonicity. The desired result follows.
\end{proof}
\par
To solve \eqref{orgi-3}, we
apply the Lagrange dual method to remove the constraint $\int_0^1 G(z)dz-\Delta\geqslant 0$ and consider the following auxiliary problem:
\begin{equation}\label{orgi-6}
\begin{array}{rl}
  \max\limits_{ {G(\cdot)}}  & \quad  U_\Delta(\lambda, G(\cdot)):=\int_0^1 [u(W_\Delta-G(z))T'(z)+\lambda G(z)]dz -\lambda\Delta, \\%[3pt]
  \mathrm{s.t.} & \quad  G(\cdot)\in {\mathbb{G}}.
\end{array}
\end{equation}
\par
The existence of the optimal solutions to \eqref{orgi-3} and \eqref{orgi-6} (for each given $\lambda \in \mathbb{R}^+$) is established in Appendix B, while the uniqueness is straightforward  when the utility function $u$ is strictly concave.
\par
To derive the optimal solution to \eqref{orgi-3}, we first solve  \eqref{orgi-6}  to obtain an optimal solution, denoted by $\widetilde{G}_\lambda(\cdot)$. Then we determine $\lambda^* \in \mathbb{R}^+$ by binding the constraint $\int_0^1 \widetilde{G}_{\lambda^*}(z)dz=\Delta$. A standard duality
argument then deduces that $G^*(\cdot):=\widetilde{G}_{\lambda^*}(\cdot)$ is an optimal solution to \eqref{orgi-3}.
Finally, an optimal solution to \eqref{orgi-2} is given by $R^*(z)=G^*(F_X(z))$ $\forall z\in [0,M]$ and that to (1) by $I^*(z)=z-R^*(z)$ $\forall z\in [0,M]$.
\par
So our problem boils down to solving \eqref{orgi-6}. However, in doing so the constraint that $0\leqslant G'(z)\leqslant (F^{-1}_X)'(z)$ in  ${\mathbb{G}}$ poses the major difficulty compared with Bernard et al. (2015).

%%%%%%%%%%%%%%%%%%%%%%%%%%%%
\section{Characterization of Solutions}
%%%%%%%%%%%%%%%%%%%%%%%%%%%%
\noindent
 In this section, we derive a necessary and sufficient condition for a solution to  be optimal to \eqref{orgi-6}. Assume $\widetilde{G}_\lambda(\cdot)$ solves \eqref{orgi-6} with a fixed $\lambda$. Let $G(\cdot)\in{\mathbb{G}}$ be arbitrary and fixed. For any $\varepsilon\in(0,1)$, set $G^\epsilon(\cdot)=(1-\epsilon)\widetilde{G}_\lambda(\cdot)+\epsilon G(\cdot)$. Then $G^\epsilon(\cdot)\in{\mathbb{G}}$. By the optimality of $\widetilde{G}_\lambda(\cdot)$ and the concavity of $u$, we have
\begin{eqnarray}\label{CoV0}
\nonumber
0\geqslant && \frac{1}{\varepsilon}\left\{\int_0^1 \left[ u(W_\Delta-G^\epsilon(z))T'(z)+\lambda G^\epsilon(z) \right] dz-\int_0^1 \left[u(W_\Delta-\widetilde{G}_\lambda(z))T'(z)+\lambda \widetilde G_\lambda(z)\right]dz\right\} \\ \nonumber
= &&\frac{1}{\varepsilon}\left\{\int_0^1 \left[(u(W_\Delta-G^\epsilon(z))-u(W_\Delta-\widetilde{G}_\lambda(z)))T'(z)+ \lambda (G^\epsilon(z)- \widetilde{G}_\lambda(z))\right]dz\right\}\\ \nonumber
\geqslant && \frac{1}{\varepsilon}\left\{\int_0^1 \left[(u'(W_\Delta-G^\epsilon(z)))(W_\Delta-G^\epsilon(z)-W_\Delta+\widetilde{G}_\lambda(z))T'(z)+ \lambda (G^\epsilon(z)- \widetilde{G}_\lambda(z))\right]dz\right\}\\ \nonumber
\underrightarrow{\epsilon \downarrow 0} && \int_0^1 \left[(u'(W_\Delta-\widetilde{G}_\lambda(z)))(\widetilde{G}_\lambda(z)-G(z))T'(z)+ \lambda (G(z)- \widetilde{G}_\lambda(z))\right]dz \\
= && \int_0^1 \left[u'(W_\Delta-\widetilde{G}_\lambda(z))T'(z)-\lambda \right](\widetilde{G}_\lambda(z)-G(z))dz.
\end{eqnarray}

Define
\begin{align}
N_\lambda(z):=-\int_z^1 \left[u'(W_\Delta-\widetilde{G}_\lambda(t))T'(t)-\lambda \right] dt,\quad z\in[0,1].
\end{align}

Then \eqref{CoV0} yields
\begin{align*}
0 \geqslant & \int_0^1 \left[u'(W_\Delta-\widetilde{G}_\lambda(z))T'(z)-\lambda \right](\widetilde{G}_\lambda(z)-G(z))dz
%=& \int_0^1 N_\lambda'(z)(\widetilde{G}_\lambda(z)-G(z))dz \\
%=& \int_0^1 (\widetilde{G}_\lambda(z)-G(z))dN_\lambda(z) \\
= \int_0^1 \int_0^z (\widetilde{G}_\lambda'(t)-G'(t)) dt dN_\lambda(z) \\
=& \int_0^1 \int_t^1 (\widetilde{G}_\lambda'(t)-G'(t)) dN_\lambda(z)dt
%=& -\int_0^1 N_\lambda(t)(\widetilde{G}_\lambda'(t)-G'(t)) dt
=\int_0^1 N_\lambda(t) (G'(t)-\widetilde{G}_\lambda'(t))dt,
\end{align*}
leading to
\begin{align*}
  \int_0^1 N_\lambda(z)  G'(z)dz  \leqslant \int_0^1 N_\lambda(z) \widetilde{G}_\lambda'(z) dz, \ \ \forall G(\cdot)\in {\mathbb{G}}.
\end{align*}
In other words, $\widetilde{G}_\lambda'(\cdot)$ maximizes $\int_0^1 N_\lambda(z)  G'(z)dz$ over $G(\cdot)\in {\mathbb{G}}$. Therefore, a necessary condition for $\widetilde{G}_\lambda(\cdot)$ to be optimal for \eqref{orgi-6} is
\begin{align} \label{optimal}
  \widetilde{G}_\lambda'(z) \overset{a.e.}{=} \begin{cases}
            0, &\quad  \text{  if  }          N_\lambda(z)=\int_z^1 [ \lambda- u'(W_\Delta-\widetilde{G}_\lambda(t))T'(t) ] dt < 0, \\%[3pt]
            \in [0,(F^{-1}_X)'(z)], &\quad  \text{  if  }          N_\lambda(z)=\int_z^1 [ \lambda- u'(W_\Delta-\widetilde{G}_\lambda(t))T'(t) ] dt = 0,\\%[3pt]
            (F^{-1}_X)'(z), &\quad  \text{  if  }        N_\lambda(z)=\int_z^1 [ \lambda- u'(W_\Delta-\widetilde{G}_\lambda(t))T'(t) ] dt > 0.
          \end{cases}
\end{align}
\par
It turns out that \eqref{optimal} completely characterizes the optimal solutions to \eqref{orgi-6}.
\begin{thm}\label{theorem:main}
A function $\widetilde{G}_\lambda(\cdot)$ is an optimal solution to \eqref{orgi-6} if and only if $\widetilde{G}_\lambda(\cdot)\in {\mathbb{G}}$ and $\widetilde{G}_\lambda(\cdot)$ satisfies \eqref{optimal}.
\end{thm}
\begin{proof}
We only need to prove the "if" part. For any feasible $G(\cdot)$ in ${\mathbb{G}}$, we have
\begin{align*}
& U_\Delta(\lambda,\widetilde{G}_\lambda(\cdot))-U_\Delta(\lambda,G(\cdot)) \\
%=& \int_0^1 [u(W_\Delta-\widetilde{G}_\lambda(z))T'(z)+ \lambda \widetilde{G}_\lambda(z)]dz - \int_0^1 [u(W_\Delta-{G}(z))T'(z)+ \lambda {G}(z)]dz\\
=& \int_0^1 [u(W_\Delta-\widetilde{G}_\lambda(z))-u(W_\Delta-{G}(z))]T'(z) dz +  \int_0^1 \lambda (\widetilde{G}_\lambda(z)- {G}(z))dz \\
\geqslant & \int_0^1 u'(W_\Delta-\widetilde{G}_\lambda(z))({G}(z) - \widetilde{G}_\lambda(z))T'(z) dz -  \int_0^1 \lambda ({G}(z) - \widetilde{G}_\lambda(z)) dz \\
%=& \int_0^1  (u'(W_\Delta-\widetilde{G}_\lambda(z))T'(z)- \lambda)({G}(z) - \widetilde{G}_\lambda(z)) dz \\
=& \int_0^1   N_\lambda'(z)({G}(z) - \widetilde{G}_\lambda(z)) dz
= \int_0^1 N_\lambda(t) (\widetilde{G}_\lambda'(t) - {G}'(t))dt \geqslant 0.
\end{align*}
%where the last inequality follows from the fact that $\widetilde{G}_\lambda(\cdot)$ satisfies \eqref{optimal}.
Hence, $\widetilde{G}_\lambda(\cdot)$ is optimal for \eqref{orgi-6}.
\end{proof}
\par
The above theorem establishes a general characterization result for the optimal solutions of \eqref{orgi-6}. This result, however, is only implicit as an optimal  $\widetilde{G}_\lambda(\cdot)$ appears on both sides of \eqref{optimal}. Moreover, the derivative of $\widetilde{G}_\lambda(z)$ is undetermined when $N_\lambda(z)=0$. In the next two sections, we will apply this general result to derive the solutions.

\section{Model with Yaari's Dual Criterion}
\noindent
When $u(x)\equiv x$, the corresponding $V^{rdu}$ reduces to the so-called Yaari's dual
criterion (Yaari 1987). In this section we solve our insurance problem with Yaari's criterion
by applying Theorem \ref{theorem:main}.  In this case, the condition \eqref{optimal} is greatly simplified. Indeed, when $u(x)\equiv x$, \eqref{optimal} reduces to \\
\begin{align} \label{equation2}
  \widetilde{G}_\lambda'(z) \overset{a.e.}{=} \begin{cases}
            0, &\quad  \text{  if  }          \int_z^1 (\lambda - T'(t) ) dt=\lambda (1-z)-(1-T(z)) < 0, \\%[3pt]
                        \in [0,(F^{-1}_X)'(z)], &\quad  \text{  if  }        \int_z^1 (\lambda - T'(t) ) dt=\lambda (1-z)-(1-T(z)) = 0, \\%[3pt]
                                 (F^{-1}_X)'(z), &\quad  \text{  if  }      \int_z^1 (\lambda - T'(t) ) dt=\lambda (1-z)-(1-T(z)) > 0  .
          \end{cases}
\end{align}
\par
It should be noted that although $u(x)\equiv x$ is not strictly concave here,  the uniqueness of optimal solution to \eqref{orgi-6} is implied by the characterizing condition \eqref{equation2}.
\par
To apply \eqref{equation2}, we need to compare  $\lambda$ and $\frac{1-T(z)}{1-z}$.
Define $f(z):=\frac{1-T(z)}{1-z}$, $z\in[0,1)$.
\begin{lemma}\label{desp}
The function $f(\cdot)$ is a continuous function on $[0,1)$. Moreover, under Assumption \ref{assump:23}, there exists a unique $a \in (0,b)$ such that $f(\cdot)$ is strictly decreasing on $[0,a]$ and strictly increasing on $[a,1)$.
\end{lemma}
\begin{proof}
We have $f'(z)=\frac{(1-T(z))-T'(z)(1-z)}{(1-z)^2}$, $z\in[0,1)$. Let $p(z):=(1-T(z))-T'(z)(1-z)$. Then $p'(z)=-T'(z)+T'(z)-T''(z)(1-z)=-T''(z)(1-z)$. It follows from Assumption \ref{assump:23} that $p'(z) > 0$ for $z \in (0,b)$ and $p'(z) < 0$  for $z \in (b,1)$. Moreover, $p(0+)=1-T'(0+)<0$, $p(b)=(1-T(b))-T'(b)(1-b)=\left(\frac{1-T(b)}{1-b}-T'(b)\right)(1-b)>0$, and $p(1-)=\lim_{z\uparrow1} (\frac{1-T(z)}{1-z}-T'(z))(1-z)\geqslant 0$ (noting $T(\cdot)$ is strictly convex on $[b,1]$). So, there exists $a \in (0,b)$ such that $p(z)<0$  for $z \in [0,a)$ and $p(z)>0$ for $z \in (a,1)$. The desired result follows.
\end{proof}
\par
Clearly, $f(0)=1, f(1-)=+\infty$.
Set $\widehat{\lambda}:=f(a) <f(0)=1$. From the proof of Lemma \ref{desp}, $a$ is determined
 by $T'(a)=\frac{1-T(a)}{1-a}$. Let $c \in (a,1]$ be the unique scalar such that $f(c)=1$ or $T(c)=c$. See Figure 1 for the locations of the points $a$ and $c$.
\par
Now, we proceed by considering three cases based on the value of $\lambda$.
\noindent
\begin{case}
$\lambda\leqslant \widehat{\lambda}$.
\end{case}
\par
In this case, $N_\lambda(z)=(1-z)(\lambda - f(z)) < 0 \quad \forall z\in [0,a)\cup(a,1]$. It then follows from $\eqref{equation2}$ that $\widetilde{G}_\lambda '(z) \overset{a.e.}{=} 0$; hence $\widetilde{G}_\lambda(z) = 0$ $\forall \ z\in [0,1]$.
Thus the corresponding retention $\widetilde{R}_\lambda(z)=0 \ \forall z \in [0,M]$ and indemnity  $\widetilde{I}_\lambda(z)=z$ $\forall z\in [0,M]$, namely, the optimal contract is a full insurance contract.

\begin{case}
$ \widehat{\lambda} < \lambda < 1 $.
\end{case}
\par
By Lemma \ref{desp}, there exist unique $x_0\in(0,a)$ and $y_0\in(a,c)$ such that $f(x_0)=f(y_0)=\lambda $. Accordingly, we have
\begin{align*} \label{diff}
   N_\lambda(z)=\begin{cases}
            < 0, &\quad  \text{  if  }           0< z < x_0,\\
            > 0, &\quad  \text{  if  }              x_0<z<y_0,\\
            < 0, &\quad  \text{  if  }            y_0 < z < 1.
          \end{cases}
\end{align*}
\par
Hence, \eqref{equation2} leads to  the following function:
\begin{align}
   \widetilde{G}_\lambda (z)=\begin{cases}
            0, &\quad  \text{  if  }           0\leqslant z < x_0, \\%[3pt]
            F^{-1}_X(z)-F^{-1}_X(x_0) , &\quad  \text{  if  }              x_0\leqslant z< y_0, \\
           F^{-1}_X(y_0)-F^{-1}_X(x_0) ,&\quad  \text{  if  }            y_0\leqslant z \leqslant 1.
          \end{cases}
\end{align}
The corresponding retention and indemnity functions are, respectively,
\begin{align*}
   \widetilde{R}_\lambda (z)\equiv \widetilde{G}_\lambda (F_X(z))=\begin{cases}
            0, &\quad  \text{  if  }             0\leqslant z < F^{-1}_X(x_0) ,\\%[3pt]
                       z-F^{-1}_X(x_0), &\quad  \text{  if  }           F^{-1}_X(x_0)\leqslant z < F^{-1}_X(y_0) ,\\
            F^{-1}_X(y_0)-F^{-1}_X(x_0),&\quad  \text{  if  }          F^{-1}_X(y_0)\leqslant z \leqslant M,
          \end{cases}
\end{align*}
and
\begin{align}
   \widetilde{I}_\lambda (z)\equiv z-\widetilde{R}_\lambda (z)=\begin{cases}
            z, &\quad  \text{  if  }           0\leqslant z < F^{-1}_X(x_0) ,\\%[3pt]
             F^{-1}_X(x_0) , &\quad  \text{  if  }        F^{-1}_X(x_0)\leqslant z < F^{-1}_X(y_0) ,\\
            z-F^{-1}_X(y_0)+F^{-1}_X(x_0),&\quad  \text{  if  }             F^{-1}_X(y_0)\leqslant z \leqslant M.
          \end{cases}
\end{align}
\par
The corresponding indemnity function is illustrated by Figure 2. Qualitatively, the insurance covers not only large losses (when $z\geqslant F^{-1}_X(y_0)$) but also {\it small} losses (when $z<F^{-1}_X(x_0)$), and the compensation is a constant for the median
range of losses. We term such a contract a {\it threefold} one. The need
for small loss coverage along with its connection to the probability weighting are
 amply discussed in Bernard et al. (2015). However, in Bernard et al. (2015) the optimal indemnity is strictly decreasing in some ranges of the losses. Such a contract may incentivize  the insured to hide partial losses in order to get more compensations. In contrast, both our indemnity and retention are increasing functions of the loss, which
  will rule out this sort of moral hazard.
\begin{figure}[htb]
  \centering
      \includegraphics[width=4 in]{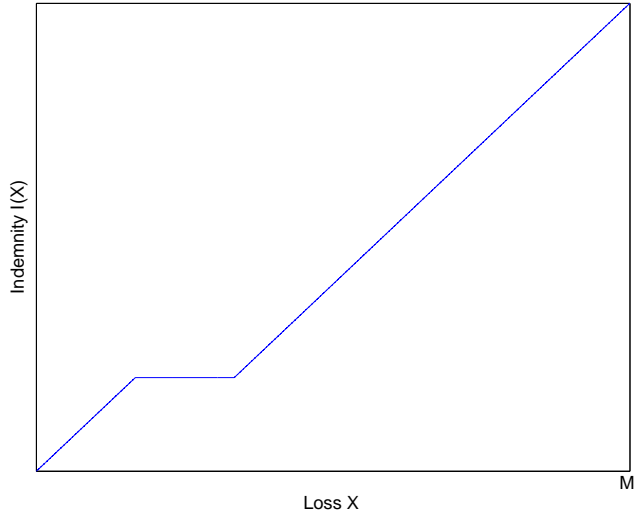}
      \caption{A threefold contract.}
\end{figure}
\begin{case}
$1 \leqslant \lambda < +\infty $
\end{case}
\par
By Lemma \ref{desp}, there exists a unique $z_0 \in [c,1]$ such that $f(z_0)=\lambda$. Thus
\begin{align*} \label{differ}
   N_\lambda(z)=\begin{cases}
            >0, &\quad  \text{  if  }              0<z<z_0,\\
            <0,&\quad  \text{  if  }            z_0 < z < 1.
          \end{cases}
\end{align*}
\noindent
By \eqref{equation2}, we have
\begin{align}
   \widetilde{G}_\lambda (z)=\begin{cases}
            F^{-1}_X(z), &\quad  \text{  if  }           0\leqslant z < z_0,\\%[3pt]
             F^{-1}_X(z_0),&\quad  \text{  if  }            z_0\leqslant z \leqslant 1.
          \end{cases}
\end{align}
So
%\begin{align*}
%  \widetilde{R}_\lambda (z)=\begin{cases}
%              z, &\quad  \text{  if  }         0\leqslant z < F^{-1}_X(z_0), \\%[3pt]
%            F^{-1}_X(z_0),&\quad  \text{  if  }             F^{-1}_X(z_0)\leqslant z \leqslant M,
%          \end{cases}
%\end{align*}
%and
\begin{align}
   \widetilde{I}_\lambda (z)\equiv z-\widetilde{R}_\lambda (z)=\begin{cases}
            0, &\quad  \text{  if  }            0\leqslant z < F^{-1}_X(z_0), \\%[3pt]
            z-F^{-1}_X(z_0),&\quad  \text{  if  }             F^{-1}_X(z_0)\leqslant z \leqslant M .
          \end{cases}
\end{align}
\par
This contract is a standard deductible contract in which only losses above a deductible point will be covered.
\par
We now summarize our results. Define
\begin{align}
   \bar{G}(z)=\begin{cases}
            F^{-1}_X(z), &\quad  \text{  if  }           0\leqslant z < c,\\%[3pt]
             F^{-1}_X(c),&\quad  \text{  if  }            c\leqslant z \leqslant 1,
          \end{cases}
\end{align}
and let $K_c:=\int_0^1 \bar{G}(z) dz$. Clearly $K_c\leqslant \int_0^1 F^{-1}_X(z) dz=E[X]$.
\begin{prop}\label{prop:yaari}
Under Yaari's criterion, $u(x)\equiv x$, and Assumptions \ref{assump:21} and \ref{assump:23}, we have the following conclusions:
\begin{enumerate}[(i)]
\item If $\Delta = 0$, then the optimal solution to \eqref{orgi-3} is ${G}^*(z)=0, 0\leqslant z\leqslant 1$.
\item  If $0<\Delta <K_c$, then the optimal solution to \eqref{orgi-3} is
\begin{align}\label{Gstar1}
   {G}^*(z)=\begin{cases}
            0, &\quad  \text{  if  }           0\leqslant z < d, \\%[3pt]
            F^{-1}_X(z)-F^{-1}_X(d) , &\quad  \text{  if  }              d\leqslant z< e, \\
           F^{-1}_X(e)-F^{-1}_X(d) ,&\quad  \text{  if  }            e\leqslant z \leqslant 1,
          \end{cases}
\end{align}
where $(d,e)$ is the unique pair satisfying $0\leqslant d<a<e\leqslant c,f(d)=f(e)$ and $\int_0^1 {G}^*(z) dz = \Delta$.
\item  If $K_c \leqslant \Delta \leqslant E[X] $, then the optimal solution to \eqref{orgi-3} is
\begin{align}
   {G}^*(z)=\begin{cases}
            F^{-1}_X(z), &\quad  \text{  if  }           0\leqslant z < q, \\%[3pt]
             F^{-1}_X(q),&\quad  \text{  if  }            q\leqslant z \leqslant 1,
          \end{cases}
\end{align}
where $q$ the unique number satisfying $c \leqslant q$ and  $\int_0^1 {G}^*(z) dz = \Delta$.
\end{enumerate}
\end{prop}
\begin{proof}
\begin{enumerate}[(i)]
\item When $\Delta=0$, the optimal solution to \eqref{orgi-3} is trivially ${G}^*(z)=0, 0\leqslant z\leqslant 1$.
\item When $0<\Delta <K_c$, there exists a unique pair $(d,e)$ such that
$0\leqslant d<a<e\leqslant c,f(d)=f(e)$ and $\int_0^1 {G}^*(z) dz = \Delta$ where $G^*$ is defined by \eqref{Gstar1}.
The existence of this pair follows from the condition that $\Delta<K_c$ and the definition of $K_c$, whereas the uniqueness comes from the requirement that $f(d)=f(e)$ and $\int_0^1 G^*(z) dz = \Delta$. Letting $\lambda\equiv \lambda_\Delta:=f(d)$, it is easy to show that $G^*(\cdot)$ satisfies \eqref{equation2} under $\lambda$, corresponding to the aforementioned Case 4.2.
This implies that $G^*(\cdot)$ is optimal for \eqref{orgi-3} under $\Delta$.
\item When $K_c \leqslant \Delta \leqslant E[X] $, a case corresponding to Case 4.3, the desired result can be derived similarly as in (ii).
\end{enumerate}
The proof is completed.
\end{proof}
\par
We are now in the position to state our main result in terms of the premium $\pi$ and the indemnity function $I(\cdot)$. Denote by $\pi_c:=(1+\rho)(E[X]-K_c)$.
\begin{thm}
Under Yaari's criterion, $u(x)\equiv x$, and  Assumptions \ref{assump:21} and \ref{assump:23}, the optimal indemnity function $I^*(\cdot)$ to Problem \eqref{orgi} is given as
\begin{enumerate}[(i)]
\item If $\pi\geqslant (1+\rho)E[X]$, then ${I}^*(z)=z$ $\forall z\in [0,M]$.
\item If $\pi_c<\pi <(1+\rho)E[X]$, then
\begin{align}
   {I}^*(z)=\begin{cases}
            z, &\quad  \text{  if  }           0\leqslant z < F^{-1}_X(d), \\%[3pt]
             F^{-1}_X(d) , &\quad  \text{  if  }        F^{-1}_X(d)\leqslant z < F^{-1}_X(e), \\
            z-F^{-1}_X(e)+F^{-1}_X(d),&\quad  \text{  if  }             F^{-1}_X(e)\leqslant z \leqslant M,
          \end{cases}
\end{align}
where $(d,e)$ is the unique pair satisfying  $0\leqslant d<a<e\leqslant c,f(d)=f(e)$ and $E[I^*(X)]=\frac{\pi}{1+\rho}$.
\item If $ 0 \leqslant \pi \leqslant \pi_c$, then
\begin{align}
   {I}^*(z)=\begin{cases}
            0, &\quad  \text{  if  }            0\leqslant z < F^{-1}_X(q), \\%[3pt]
            z-F^{-1}_X(q),&\quad  \text{  if  }             F^{-1}_X(q)\leqslant z\leqslant M,
          \end{cases}
\end{align}
where $q$ is the unique scalar satisfying $c \leqslant q$ and $E[I^*(X)]=\frac{\pi}{1+\rho}$.
\end{enumerate}
\end{thm}
\begin{proof}
Since $\Delta=E[X]-\frac{\pi}{1+\rho}$, the constraint $E[R(X)]\equiv \int_0^1 {G}_{R(X)}(z) dz = \Delta$ is equivalent to that $E[I(X)]=E[X]-E[R(X)]= \frac{\pi}{1+\rho}$. The desired result is then a direct consequence of Proposition \ref{prop:yaari}.
\end{proof}
\par
The economic interpretation of this result is clear. When the premium is small ($ 0 \leqslant \pi \leqslant \pi_c$), the insurance only compensates large losses in excess of certain amount. When the premium is in middle range ($\pi_c<\pi <(1+\rho)E[X]$), the contract is a threefold one, covering both small and large losses. When the premium is sufficiently large ($\pi\geqslant (1+\rho)E[X]$), it is a full coverage.
\par
It is interesting to  investigate the comparative statics of the point $\pi_c$ (in terms of $c$) that triggers the coverage for small losses. In fact, as $K_c=\int_0^c F^{-1}_X(z) dz + F^{-1}_X(c)(1-c) $, we have $\frac{\partial K_c}{\partial c}=(1-c)(F^{-1}_X)'(c).$ However, $\pi_c=(1+\rho)(E[X]-K_c)$; hence $\frac{\partial \pi_c}{\partial c}=(1+\rho)(c-1)(F^{-1}_X)'(c) <0.$ This implies that the insurer is more willing to be protected against small losses if his weighting function has a  bigger $c$. This is consistent with the fact that a bigger $c$ renders a larger concave domain of the probability weighting that overweighs small losses (refer to Figure 1).

 %%%%%%%%%%%%%%%%%%%%%%%%%%%%
\section{Model with the RDU Criterion }
%%%%%%%%%%%%%%%%%%%%%%%%%%%%
\noindent
In this section we study the general RDU model in which the utility function is strictly concave. Compared with the Yaari model, solving the corresponding insurance problem calls for a more delicate analysis.
\par
For any twice differentiable function $f$ with $f'(x)\neq 0$, define its {\it Arrow-Pratt measure of absolute risk aversion} $A_f(x):=-\frac{f''(x)}{f'(x)}$.
We now introduce the following assumptions.
\begin{assump}\label{scu}
(Strictly Concave Utility) The  utility function $u:\mathbb{R}^+\mapsto \mathbb{R}^+$ is strictly increasing and  twice differentiable. Furthermore, $u'$ is strictly decreasing.
\end{assump}
\begin{assump}\label{ap}
\begin{enumerate}[(i)]
\item The function  $A_u(z)$ is decreasing on $(0,\infty)$.
\item $A_T(z)> A_u(W-F^{-1}_X(z))(F^{-1}_X)'(z)$ $\forall z \in (0,a]$.
\end{enumerate}
\end{assump}
\par
Assumption \ref{scu} is to replace Assumption \ref{assump:22}, ensuring a genuine RDU criterion. Assumption \ref{ap}-(i) requires that the absolute risk aversion measure of the utility function $u$ be decreasing, which holds true for many frequently used utility functions including logarithmic, power and exponential utilities. In general, experimental and empirical evidences are consistent with the decreasing absolute risk aversion; see e.g. Friend, Irwin and Blume, Marshall (1975). On the other hand, $A_T(z)$, $z \in (0,a]$, measures the level of probability weighting for small losses. The economical interpretation of Assumption \ref{ap}-(ii) is, therefore, that the degree of the insured's concern for small losses is sufficiently large relative to the absolute risk aversion of the utility function.
Note that Assumption \ref{ap}-(ii) is automatically satisfied when $F^{-1}_X(z)=0$, $\forall z \in [0,a]$, which is equivalent to $\mathbb{P}(X=0)\geqslant a$. In practice, $\mathbb{P}(X=0)\geqslant 0.5$ is a plausible assumption for many insurance products such as automobile and house insurance. On the other hand, $a$ is very small for many
commonly used inverse-$S$ shaped weighting functions. Take Tversky and Kahneman's weighting function \eqref{TK-distortion} as an example, $a\approx 0.013$
when $\theta=0.3$,  $a\approx 0.07$ when $\theta=0.5$, and $a\approx 0.166$ when $\theta=0.8$. In these cases, Assumption \ref{ap}-(ii) holds automatically.  

\par
Problem \eqref{orgi-3} has trivial solutions in the following two cases. When $\Delta=0$,  the optimal solution is $G^*(z)=0$ $\forall z\in [0,1]$, corresponding to a full coverage. When $\Delta=E[X]$, the optimal solution is $G^*(z)=F^{-1}_X(z)$ $\forall z\in [0,1]$ as it is the only feasible solution, corresponding to no coverage.
\par
So we are interested in only the case  $0<\Delta<E[X]$. It follows from Proposition \ref{prop-dual} in Appendix C that there exists $\lambda^*$ such that $\widetilde{G}_{\lambda^*}(\cdot)$ is optimal solution to \eqref{orgi-6} under $\lambda^*$ and $\int_0^1 \widetilde{G}_{\lambda^*}(z) dz = \Delta$. Furthermore, recall that we have proved that \eqref{orgi-6} has a unique solution when $u$ is strictly concave and \eqref{optimal} provides the necessary and sufficient condition for the optimal solution.

\begin{lemma}\label{lemma:51}
For any $G(\cdot)\in {\mathbb{G}}$, if there exists $z\in(0,1)$ such that $\lambda- u'(W_\Delta-G(z))T'(z)=\int_z^1 [ \lambda- u'(W_\Delta-G(t))T'(t) ] dt = 0$, then $z\leqslant a$.
\end{lemma}
\begin{proof}
From $\lambda- u'(W_\Delta-G(z))T'(z)=0$, it follows
$u'(W_\Delta-G(z))=\frac{\lambda}{T'(z)}$. Hence, if $z>a$, then
\begin{eqnarray*}
0= && \int_z^1 \left[ \lambda- u'(W_\Delta-G(t))T'(t) \right] dt \\
\leqslant && \int_z^1 \left[ \lambda- u'(W_\Delta-G(z))T'(t) \right] dt
=  \frac{\lambda}{T'(z)}(1-z) \left[T'(z)-\frac{1-T(z)}{1-z}\right]<0,
\end{eqnarray*}
where the last inequality is due to Lemma \ref{lemma:A1}-(i) in Appendix A, noting $z>a$. This is a contradiction.
\end{proof}

\begin{lemma}\label{lemma:52}
Under Assumption \ref{ap}, for any \ $G(\cdot) \in {\mathbb{G}}$, $u'(W_\Delta-G(z))T'(z)$ is a strictly decreasing function of $z$ on $[0,a]$.
\end{lemma}
\begin{proof}
Noting $W_\Delta\geqslant W$, it follows from Assumption \ref{ap} that $A_T(z)> A_u(W-F^{-1}_X(z))(F^{-1}_X)'(z) \geqslant A_u(W_\Delta-F^{-1}_X(z))(F^{-1}_X)'(z) $ $\forall z \in (0,a]$. Now, we compute  the partial derivative of $u'(W_\Delta-G(z))T'(z)$ with respect to $z\in [0,a]$:
\begin{eqnarray*}
&& \frac{\partial}{\partial z}(u'(W_\Delta-G(z))T'(z)) \\
=&& -u''(W_\Delta-G(z))G'(z)T'(z)+u'(W_\Delta-G(z))T''(z) \\
=&& u'(W_\Delta-G(z))T'(z)\left[A_u(W_\Delta-G(z))G'(z)-A_T(z)\right] \\
<&& u'(W_\Delta-G(z))T'(z)\left[A_u(W_\Delta-G(z))G'(z)-A_u(W_\Delta-F^{-1}_X(z))(F^{-1}_X)'(z)\right]  \\
\leqslant && u'(W_\Delta-G(z))T'(z)\left[A_u(W_\Delta-G(z))(F^{-1}_X)'(z)-A_u(W_\Delta-F^{-1}_X(z))(F^{-1}_X)'(z)\right]  \\
\leqslant && u'(W_\Delta-G(z))T'(z)\left[A_u(W_\Delta-F^{-1}_X(z))(F^{-1}_X)'(z)-A_u(W_\Delta-F^{-1}_X(z))(F^{-1}_X)'(z)\right] \\
=&& 0.
\end{eqnarray*}
The proof is complete.
\end{proof}
\par
Now, for any $\lambda \leqslant \widehat{\lambda} u'(W_\Delta)$, we have
\begin{eqnarray*}
\int_z^1 \left[ \lambda- u'(W_\Delta-\widetilde{G}_\lambda(t))T'(t) \right] dt
\leqslant && \int_z^1 \left[ \widehat{\lambda}u'(W_\Delta)- u'(W_\Delta)T'(t) \right] dt \\
= && u'(W_\Delta) \int_z^1 [\widehat{\lambda}-T'(t)] dt
= u'(W_\Delta)(1-z)\left[\widehat{\lambda}-\frac{1-T(z)}{1-z}\right]< 0,
\end{eqnarray*}
where the last inequality is due to Lemma \ref{desp}.
Hence $\widetilde{G}_\lambda(z)=0$ $\forall z \in [0,1]$ is the only solution satisfying \eqref{optimal}. However, $\int_0^1 \widetilde{G}_{\lambda}(z) dz = 0<\Delta$, a contradiction. Therefore, only when $\lambda > \widehat{\lambda}u'(W_\Delta)$ is it possible for \eqref{optimal} to hold.

\par
Fixing $\lambda > \widehat{\lambda}u'(W_\Delta)$, we now analyze the shape of the function $\widetilde{G}_\lambda(\cdot)$ that satisfies \eqref{optimal}. Assume that $\widetilde{G}_\lambda(1)=k<W_\Delta$. We have $N_\lambda(1)=0$ and $\lambda- u'(W_\Delta-k)T'(1-)<0$ since $T'(1-)=+\infty$. So, $\widetilde{G}_\lambda'(z)=0$ when $z$ is close to 1 since $N_\lambda(z)<0$ for such $z$. Hence, $\widetilde{G}_\lambda(z)\equiv k \ \forall z\in [z_1,1]$ for some $z_1 \in [0,1)$, at which $N_\lambda(z_1)=0$ and $N_\lambda(z)<0$ for $\forall z \in (z_1,1)$. Next, we
consider three cases respectively depending on the value of $k$.
\par
(A) If $k>W_\Delta-(u')^{-1}(\frac{\lambda}{\widehat{\lambda}})$ (i.e. $\lambda<\widehat{\lambda}u'(W_\Delta-k)$), then we have, $\forall z\in [0,1)$
\begin{eqnarray*}\label{z1z1}
\int_{z}^1 \left[ \lambda- u'(W_\Delta-k)T'(t) \right] dt
< && \int_z^1 \left[ \widehat{\lambda}u'(W_\Delta-k)- u'(W_\Delta-k)T'(t) \right] dt \\
%= && u'(W_\Delta-k) \int_z^1 [\widehat{\lambda}-T(t)] dt \\
= && u'(W_\Delta-k)(1-z)\left[\widehat{\lambda}-\frac{1-T(z)}{1-z}\right]
\leqslant  0.
\end{eqnarray*}
It then follows from \eqref{optimal} that $\widetilde{G}_\lambda(z)\equiv k=\widetilde{G}_\lambda(0)=0$ . However, $0=k>W_\Delta-(u')^{-1}(\frac{\lambda}{\widehat{\lambda}})$, or $\lambda \leqslant \widehat{\lambda}u'(W_\Delta)$,  leading to a contradiction. So, this case  in fact will not take place.
\par
(B) If $k=W_\Delta-(u')^{-1}(\frac{\lambda}{\widehat{\lambda}})$, then $z_1$ should be $a$. This is because $\int_{a}^1 [ \lambda- u'(W_\Delta-k)T'(t) ] dt=0$ and $\int_{z}^1 [ \lambda- u'(W_\Delta-k)T'(t) ] dt=\frac{\lambda}{\widehat{\lambda}}(1-z)(\widehat{\lambda}-\frac{1-T(z)}{1-z})<0$ for $z\in (a,1)$ by Lemma \ref{desp}. %\{{\bf Zhou: I cannot follow here. Zhuang: I have addressed this problem and made it more clearly.}}
Moreover, $\lambda- u'(W_\Delta-k)T'(a)$=0. By Lemma \ref{lemma:52}, $\lambda- u'(W_\Delta-\widetilde{G}_\lambda(z))T'(z)$ strictly increases with respect to $z\in [0,a]$. It follows that $\lambda- u'(W_\Delta-\widetilde{G}_\lambda(z))T'(z)<0$ for $z\in [0,a)$. Then \eqref{optimal} implies $\widetilde{G}_\lambda'(z)=0$ for $z\in(0,a)$. As a result, $k=\widetilde{G}_\lambda(a)=\widetilde{G}_\lambda(0)=0$, or $\lambda=\widehat{\lambda}u'(W_\Delta)$, which is a contradiction. So, again, this case will not occur.
\par
(C) If $k<W_\Delta-(u')^{-1}(\frac{\lambda}{\widehat{\lambda}})$, then $z_1\in (a,1)$ exists. By Lemma \ref{lemma:51}, we have $\lambda- u'(W_\Delta-k)T'(z_1)>0$. Hence, there may or may not exist $z_2\in (0,1)$ such that $N_\lambda(z_2)=0$ and $N_\lambda(z)>0$ for $z\in(z_2,z_1)$. We now discuss four subcases depending on the existence and location of $z_2$.
\par
(C.1) If $z_2$ does not exist or $z_2=0$ (i.e. $N_\lambda(z)>0$ for $z\in(0,z_1)$), then by \eqref{optimal}, $\widetilde{G}_\lambda'(z)=(F^{-1}_X)'(z)$ for $z\in(0,z_1)$. Combined with the fact that $\widetilde{G}_\lambda(0)=0$, we have:
\begin{align*}
   \widetilde{G}_\lambda(z)=\begin{cases}
            F^{-1}_X(z), &\quad  \text{  if  }           0\leqslant z < z_1,\\
             F^{-1}_X(z_1),&\quad  \text{  if  }            z_1 \leqslant  z \leqslant 1.
          \end{cases}
\end{align*}
This corresponds to a deductible contract.
\par
(C.2) If $z_2$ exists and $z_2 \in (0,a]$, then $\widetilde{G}_\lambda'(z)=(F^{-1}_X)'(z)$ for $z\in(z_2,z_1)$ in view of \eqref{optimal}. Combining the property of $z_1$ and $z_2$, we deduce $\lambda- u'(W_\Delta-\widetilde{G}_\lambda(z_2))T'(z_2)\leqslant 0$. Then, using Lemma \ref{lemma:52}, we have $\lambda- u'(W_\Delta-\widetilde{G}_\lambda(z))T'(z)< 0$ for $z\in [0,z_2)$. It follows from  \eqref{optimal} that $\widetilde{G}_\lambda'(z)=0$ for $z\in(0,z_2)$. In this case, we can express $\widetilde{G}_\lambda(\cdot)$ as follows
\begin{align*}
   \widetilde{G}_\lambda(z)=\begin{cases}
            0, &\quad  \text{  if  }              0 \leqslant z< z_2,\\
            F^{-1}_X(z)-F^{-1}_X(z_2), &\quad  \text{  if  }              z_2 \leqslant z< z_1,\\
            F^{-1}_X(z_1)-F^{-1}_X(z_2),&\quad  \text{  if  }            z_1 \leqslant z \leqslant 1.
          \end{cases}
\end{align*}
This is the threefold contract, depicted in Figure 2.
\par
(C.3) If $z_2$ exists and $z_2 \in (b,1)$ (recall that $b$ is the turning point where the weighting function $T(\cdot)$ changes from being concave to convex), then a similar analysis as in Case C.2 shows that $\lambda- u'(W_\Delta-\widetilde{G}_\lambda(z_1))T'(z_1)>0$ and $\lambda- u'(W_\Delta-\widetilde{G}_\lambda(z_2))T'(z_2)<0$. This means $u'(W_\Delta-\widetilde{G}_\lambda(z_2))T'(z_2)>u'(W_\Delta-\widetilde{G}_\lambda(z_1))T'(z_1)$. However, $u'(W_\Delta-\widetilde{G}_\lambda(z_1)) \geqslant u'(W_\Delta-\widetilde{G}_\lambda(z_2))>0$ and $T'(z_1)>T'(z_2)>0$, which is a contradiction. So, this case is not feasible.
\par
(C.4) If $z_2$ exists and $z_2 \in (a,b]$, then $\lambda- u'(W_\Delta-\widetilde{G}_\lambda(z_2))T'(z_2)<0$. We prove $\widetilde{G}_\lambda(z)\equiv \widetilde{G}_\lambda(z_2) \ \forall z\in [0,z_2]$. In fact, if it is false, then there exists $z_3$ such that $\int_{z_3}^{z_2} [\lambda-u'(W_\Delta-\widetilde{G}_\lambda(z_2))T'(t)] dt=0$ and $\int_{z}^{z_2} [\lambda-u'(W_\Delta-\widetilde{G}_\lambda(z_2))T'(t)] dt<0$ for $z\in (z_3,z_2)$. However, $\lambda-u'(W_\Delta-\widetilde{G}_\lambda(z_2))T'(z)<\lambda-u'(W_\Delta-\widetilde{G}_\lambda(z_2))T'(z_2)<0$ for $z\in (z_3,z_2)$ since $z_2 \in (a,b]$. So, $\int_{z_3}^{z_2} [\lambda-u'(W_\Delta-\widetilde{G}_\lambda(z_2))T'(t)] dt<(\lambda-u'(W_\Delta-\widetilde{G}_\lambda(z_2))T'(z_2))(z_2-z_3)<0$, arriving at a contradiction. Therefore, $k=F^{-1}_X(z_1)-F^{-1}_X(z_2)$. From $\int_{z_1}^1 [ \lambda- u'(W_\Delta-k)T'(t) ] dt=0 $, it follows $\lambda=u'(W_\Delta-k)\frac{1-T(z_1)}{1-z_1}=u'(W_\Delta+F^{-1}_X(z_2)-F^{-1}_X(z_1))\frac{1-T(z_1)}{1-z_1}$. However,
\begin{eqnarray*}
&&  \int_{z_2}^{z_1} \left[ u'(W_\Delta+F^{-1}_X(z_2)-F^{-1}_X(z_1))\frac{1-T(z_1)}{1-z_1}- u'(W_\Delta+F^{-1}_X(z_2)-F^{-1}_X(t))T'(t) \right] dt \\
> &&   \int_{z_2}^{z_1} \left[ u'(W_\Delta+F^{-1}_X(z_2)-F^{-1}_X(z_1))\frac{1-T(z_1)}{1-z_1}- u'(W_\Delta+F^{-1}_X(z_2)-F^{-1}_X(z_1))T'(t) \right] dt \\
%=&& u'(W_\Delta+F^{-1}_X(z_2)-F^{-1}_X(z_1)) \int_{z_2}^{z_1} \left[ \frac{1-T(z_1)}{1-z_1}- T'(t) \right] dt \\
=&& u'(W_\Delta+F^{-1}_X(z_2)-F^{-1}_X(z_1))(z_1-z_2) \left[ \frac{1-T(z_1)}{1-z_1}- \frac{T(z_1)-T(z_2)}{z_1-z_2} \right]>0,
%>&& 0  \ \ \ ( \ by \ Lemma \ A.1 \ in \ Appendix \ A),
\end{eqnarray*}
where the last inequality follows from  Lemma \ref{lemma:A1}-(ii) in Appendix A. This is a contradiction. So, the current case will not occur either.
\par
To summarize, for any $\lambda>\widehat{\lambda} u'(W_\Delta)$, only deductible and threefold contracts are possibly optimal, stipulated in Case C.1 and Case C.2 above. Next, we investigate these two cases more closely.
\par
Define a function $h_\Delta(\cdot)$ on $[a,c]$ as follows:
\begin{align}\label{hfunc}
h_\Delta(z):=\int_0^z \left[ \frac{u'(W_\Delta-F^{-1}_X(z))(1-T(z))}{1-z}-u'(W_\Delta-F^{-1}_X(t))T'(t) \right] dt.
\end{align}
Then, by using Lemma \ref{lemma:52} and the fact that $\frac{1-T(a)}{1-a}=T'(a)$, we have
\begin{eqnarray*}
 h_\Delta(a)= && \int_0^a \left[ u'(W_\Delta-F^{-1}_X(a))T'(a)-u'(W_\Delta-F^{-1}_X(t))T'(t) \right] dt < 0.
\end{eqnarray*}
Recalling that $T(c)=c$, we have
\begin{eqnarray*}
h_\Delta(c)= && \int_0^c \left[ \frac{u'(W_\Delta-F^{-1}_X(c))(1-T(c))}{1-c}-u'(W_\Delta-F^{-1}_X(t))T'(t) \right] dt \\
=&& \int_0^c [ u'(W_\Delta-F^{-1}_X(c))-u'(W_\Delta-F^{-1}_X(t))T'(t) ] dt \\
> && u'(W_\Delta-F^{-1}_X(c))c - \int_0^c [ u'(W_\Delta-F^{-1}_X(c))T'(t) ] dt =0.
 %=&& u'(W_\Delta-F^{-1}_X(c))c - u'(W_\Delta-F^{-1}_X(c))c  \\
 %=&& 0.
\end{eqnarray*}

Moreover, we take the derivative of $h_\Delta(z)$ with respect to $z\in [a,c]$ to obtain
\begin{eqnarray*}
 h_\Delta'(z)=&& -u'(W_\Delta-F^{-1}_X(z))T'(z)+u'(W_\Delta-F^{-1}_X(z))\frac{1-T(z)}{1-z} \\
 && -u''(W_\Delta-F^{-1}_X(z))\frac{1-T(z)}{1-z}z(F^{-1}_X)'(z)+u'(W_\Delta-F^{-1}_X(z))z\frac{\frac{1-T(z)}{1-z}-T'(z)}{1-z}\\
 =&& u'(W_\Delta-F^{-1}_X(z))\left(\frac{1-T(z)}{1-z}-T'(z)\right)- u''(W_\Delta-F^{-1}_X(z))\frac{1-T(z)}{1-z}z(F^{-1}_X)'(z) \\
 && +u'(W_\Delta-F^{-1}_X(z))z\frac{\frac{1-T(z)}{1-z}-T'(z)}{1-z} >0.
\end{eqnarray*}
Hence, there exists a unique point $l_\Delta \in (a,c)$ such that $h_\Delta(l_\Delta)=0$, $h_\Delta(z)<0$ for $z\in (a,l_\Delta)$, and $h_\Delta(z)>0$ for $z\in (l_\Delta,c)$.
\par
Define
\begin{align}
   \underline{G}(z)=\begin{cases}
            F^{-1}_X(z), &\quad  \text{  if  }           0\leqslant z < l_\Delta,\\%[3pt]
             F^{-1}_X(l_\Delta),&\quad  \text{  if  }            l_\Delta\leqslant z \leqslant 1,
          \end{cases}
\end{align}
and $K_\Delta:=\int_0^1 \underline{G}(z) dz$.
\begin{prop}\label{prop:51}
If $K_\Delta \leqslant \Delta < E[X]$, then the optimal solution to \eqref{orgi-3} is
\begin{align}
   G^*(z)=\begin{cases}
            F^{-1}_X(z), &\quad  \text{  if  }           0\leqslant z < f,\\%[3pt]
             F^{-1}_X(f),&\quad  \text{  if  }            f\leqslant z \leqslant 1,
          \end{cases}
\end{align}
where $f$ is the unique scalar such that $f \geqslant l_\Delta$ and $\int_0^1 G^*(z) dz = \Delta$.
\end{prop}
\begin{proof}
The existence of $f$ follows from the monotonicity of $G^{*}$ with respesct to $f$ immediately.
Denoting $\lambda_\Delta:=u'(W_\Delta-F^{-1}_X(f))\frac{1-T(f)}{1-f}$,  we need to show that $G^*(\cdot)$ satisfies \eqref{optimal} with $\lambda=\lambda_\Delta$.
First, it is straightforward that $\int_f^1 \left[ \frac{u'(W_\Delta-F^{-1}_X(f))(1-T(f))}{1-f}-u'(W_\Delta-F^{-1}_X(f))T'(t) \right] dt = 0$. Next, we are to prove that $\int_z^f \left[ \frac{u'(W_\Delta-F^{-1}_X(f))(1-T(f))}{1-f}-u'(W_\Delta-F^{-1}_X(t))T'(t) \right] dt>0$ $\forall z \in (0,f)$. We divide the proof into two cases.
\begin{itemize}
\item If $z \in [a,f)$, then
\begin{align*}
& \int_z^f \left[ \frac{u'(W_\Delta-F^{-1}_X(f))(1-T(f))}{1-f}-u'(W_\Delta-F^{-1}_X(t))T'(t) \right] dt \\
\geqslant & \int_z^f \left[ \frac{u'(W_\Delta-F^{-1}_X(f))(1-T(f))}{1-f}-u'(W_\Delta-F^{-1}_X(f))T'(t) \right] dt \\
%=& \frac{u'(W_\Delta-F^{-1}_X(f))(1-T(f))(f-z)}{1-f}-u'(W_\Delta-F^{-1}_X(f))[T(f)-T(z)]  \\
%=& u'(W_\Delta-F^{-1}_X(f))(f-z)\left[\frac{1-T(f)}{1-f}-\frac{T(f)-T(z)}{f-z}\right] \\
=& u'(W_\Delta-F^{-1}_X(f))(f-z)\left[\frac{1-T(f)}{1-f}-\frac{1-T(z)}{1-z}\right] > 0,
\end{align*}
where the last inequality is due to Lemma \ref{desp}.
\item If $z \in (0,a)$ and $u'(W_\Delta-F^{-1}_X(z))T'(z) \leqslant \frac{u'(W_\Delta-F^{-1}_X(f))(1-T(f))}{1-f}$, then by Lemma \ref{lemma:52} and the result above, we have
\begin{align*}
& \int_z^f \left[ \frac{u'(W_\Delta-F^{-1}_X(f))(1-T(f))}{1-f}-u'(W_\Delta-F^{-1}_X(t))T'(t) \right] dt \\
=& \int_z^a \left[ \frac{u'(W_\Delta-F^{-1}_X(f))(1-T(f))}{1-f}-u'(W_\Delta-F^{-1}_X(t))T'(t) \right] dt  \\
& + \int_a^f \left[ \frac{u'(W_\Delta-F^{-1}_X(f))(1-T(f))}{1-f}-u'(W_\Delta-F^{-1}_X(t))T'(t) \right] dt \\
>& \int_z^a \left[ \frac{u'(W_\Delta-F^{-1}_X(f))(1-T(f))}{1-f}-u'(W_\Delta-F^{-1}_X(t))T'(t) \right] dt \\
\geqslant & \int_z^a \left[u'(W_\Delta-F_X^{-1}(z))T'(z)-u'(W_\Delta-F^{-1}_X(t))T'(t) \right] dt > 0.
\end{align*}
If $z \in (0,a)$ and $u'(W_\Delta-F^{-1}_X(z))T'(z) > \frac{u'(W_\Delta-F^{-1}_X(f))(1-T(f))}{1-f}$ holds, then
\begin{align*}
 & \int_z^f \left[ \frac{u'(W_\Delta-F^{-1}_X(f))(1-T(f))}{1-f}-u'(W_\Delta-F^{-1}_X(t))T'(t) \right] dt \\
=& \int_{l_\Delta}^f \left[ \frac{u'(W_\Delta-F^{-1}_X(f))(1-T(f))}{1-f}-u'(W_\Delta-F^{-1}_X(t))T'(t) \right] dt \\
  & + \int_z^{l_\Delta} \left[ \frac{u'(W_\Delta-F^{-1}_X(f))(1-T(f))}{1-f}-u'(W_\Delta-F^{-1}_X(t))T'(t) \right] dt \\
\geqslant & \int_z^{l_\Delta} \left[ \frac{u'(W_\Delta-F^{-1}_X(f))(1-T(f))}{1-f}-u'(W_\Delta-F^{-1}_X(t))T'(t) \right] dt \\
\geqslant & \int_z^{l_\Delta} \left[ \frac{u'(W_\Delta-F^{-1}_X(l_\Delta))(1-T(l_\Delta))}{1-l_\Delta}-u'(W_\Delta-F^{-1}_X(t))T'(t) \right] dt \\
= & \int_0^{l_\Delta} \left[ \frac{u'(W_\Delta-F^{-1}_X(l_\Delta))(1-T(l_\Delta))}{1-l_\Delta}-u'(W_\Delta-F^{-1}_X(t))T'(t) \right] dt \\
 & - \int_0^z \left[ \frac{u'(W_\Delta-F^{-1}_X(l_\Delta))(1-T(l_\Delta))}{1-l_\Delta}-u'(W_\Delta-F^{-1}_X(t))T'(t) \right] dt \\
=& - \int_0^z \left[ \frac{u'(W_\Delta-F^{-1}_X(l_\Delta))(1-T(l_\Delta))}{1-l_\Delta}-u'(W_\Delta-F^{-1}_X(t))T'(t) \right] dt > 0,
\end{align*}
where the last inequality is due to $$u'(W_\Delta-F^{-1}_X(z))T'(z) > \frac{u'(W_\Delta-F^{-1}_X(f))(1-T(f))}{1-f} \geqslant \frac{u'(W_\Delta-F^{-1}_X(l_\Delta))(1-T(l_\Delta))}{1-l_\Delta},$$ as $l_\Delta \leqslant f$ and the fact that $u'(W_\Delta-F^{-1}_X(z))T'(z) $ is strictly decreasing on $[0,a]$.
\end{itemize}
The claim follows now.
\end{proof}
\begin{lemma}\label{lemma:53}
If $ 0 < \Delta < K_\Delta$, then the corresponding optimal contract is not a deductible one.
\end{lemma}
\begin{proof}
There exists $\lambda^*$ such that $\widetilde{G}_{\lambda^*}(\cdot)$ satisfies \eqref{optimal} under $\lambda^*$ and $\int_0^1 \widetilde{G}_{\lambda^*}(z)dz=\Delta$ (see Appendix C). If $\widetilde{G}_{\lambda^*}(\cdot)$ corresponds to a deductible contract, then there exists  $\overline{z}$ (since $\Delta < K_\Delta$, we have $\overline{z}<l_\Delta$) such that
\begin{align}
   \widetilde{G}_{\lambda^*}(z)=\begin{cases}
            F^{-1}_X(z), &\quad  \text{  if  }           0\leqslant z < \overline{z},\\%[3pt]
             F^{-1}_X(\overline{z}),&\quad  \text{  if  }            \overline{z}\leqslant z \leqslant 1.
          \end{cases}
\end{align}
Since $\widetilde{G}_{\lambda^*}(\cdot)$ satisfies \eqref{optimal}, we have $\int_{\overline{z}}^1 [ \lambda^*- u'(W_\Delta-F^{-1}_X(\overline{z}))T'(t) ] dt = 0$, or $\lambda^*=u'(W_\Delta-F^{-1}_X(\overline{z}))\frac{1-T'(\overline{z})}{1-\overline{z}}$.
\par
On the other hand, $M(z):=\int_{z}^{\overline{z}} [ u'(W_\Delta-F^{-1}_X(\overline{z}))\frac{1-T'(\overline{z})}{1-\overline{z}}- u'(W_\Delta-F^{-1}_X(t))T'(t) ] dt \geqslant 0$  for $z\in[0,\overline{z}]$. However, by the definition of $l_\Delta$, $h_\Delta(\overline{z})\equiv M(0)=\int_0^{\overline{z}} [ \frac{u'(W_\Delta-F^{-1}_X(\overline{z}))(1-T(\overline{z}))}{1-\overline{z}}-u'(W_\Delta-F^{-1}_X(t))T'(t) ] dt < 0$ as $\overline{z}<l_\Delta$. Since $M(\cdot)$ is a continuous function, a contradiction arises.
\end{proof}
\par
It follows from Lemma \ref{lemma:53} that, if $\ 0 < \Delta < K_\Delta$, the optimal contract (which always exists) can only be a threefold one, corresponding to Case C.2.
We are now led to the following proposition.
\begin{prop}\label{prop:52}
If $ 0 < \Delta < K_\Delta$, then the optimal solution to \eqref{orgi-3} is given as
\begin{align*}
   G^*(z)=\begin{cases}
            0, &\quad  \text{  if  }              0 \leqslant z<z_2,\\
            F^{-1}_X(z)-F^{-1}_X(z_2), &\quad  \text{  if  }              z_2 \leqslant z< z_1,\\
            F^{-1}_X(z_1)-F^{-1}_X(z_2),&\quad  \text{  if  }            z_1 \leqslant z \leqslant 1,
          \end{cases}
\end{align*}
where $z_2,z_1$ are such that $z_2\leqslant a \leqslant z_1$, $\int_{z_2}^{z_1} [ \frac{u'(W_\Delta-F^{-1}_X(z_1)+F^{-1}_X(z_2))(1-T(z_1))}{1-z_1}- u'(W_\Delta-F^{-1}_X(t)+F^{-1}_X(z_2))T'(t) ] dt=0$ and $\int_0^1 G^*(z) dz=\Delta$.
\end{prop}
\begin{proof}
The conclusion is a direct consequence of Lemma \ref{lemma:53}.
\end{proof}
\par
Note that any pair $(z_2,z_1)$ satisfying the requirements in Proposition \ref{prop:52} leads to an optimal solution to \eqref{orgi-3}. Therefore such a pair $(z_2,z_1)$ is unique as the optimal solution to \eqref{orgi-3} is unique.
\par
Proposition \ref{prop:51} and Proposition \ref{prop:52} give two qualitatively distinct optimal contracts for any given $0<\Delta<E[X]$, and the two cases are divided depending on whether or not $\Delta<K_\Delta$.
However, $K_\Delta$ in general depends on $\Delta$ in an implicit and complicated way; so it is hard to compare $\Delta$ and $K_\Delta$. Nevertheless we are able to treat at least  two cases where
$A_u(z)$ is either a constant or strictly decreasing in $z$.
\par
First, assume that the utility function exhibits constant absolute risk aversion, e.g. $u(z)=1-e^{-\alpha z}$ $\forall z\in \mathbb{R}^+$. Then
it is easy to see from \eqref{hfunc} that $l_\Delta$ is independent of $\Delta$, and hence so is $K_\Delta$. In this case, denote $K\equiv K_\Delta$ and $\widehat{\pi}=(1+\rho)(E[X]-K)$. Then
we have the following result.
\begin{thm}\label{thm:54}
Assume that Assumptions \ref{assump:21}, \ref{assump:23}, and \ref{ap} hold, and that $u(\cdot)$  exhibits constant absolute risk aversion. Then the optimal indemnity function $I^*(\cdot)$ to Problem \eqref{orgi} is given as
\begin{enumerate}[(i)]
\item If $\pi=(1+\rho)E[X]$, then $I^*(z)=z$ for $z\in [0,M]$.
\item If $\widehat{\pi}<\pi <(1+\rho)E[X]$, then
\begin{align*}
   {I}^*(z)=\begin{cases}
            z, &\quad  \text{  if  }           0\leqslant z < F^{-1}_X(z_2),\\%[3pt]
             F^{-1}_X(z_2) , &\quad  \text{  if  }        F^{-1}_X(z_2)\leqslant z < F^{-1}_X(z_1),\\
            z-F^{-1}_X(z_1)+F^{-1}_X(z_2),&\quad  \text{  if  }             F^{-1}_X(z_1)\leqslant z \leqslant M,
          \end{cases}
\end{align*}
where $(z_2,z_1)$ is the unique pair satisfying $z_2\leqslant a \leqslant z_1$,
$\int_{z_2}^{z_1} [ \frac{u'(W_\Delta-F^{-1}_X(z_1)+F^{-1}_X(z_2))(1-T(z_1))}{1-z_1}- u'(W_\Delta-F^{-1}_X(t)+F^{-1}_X(z_2))T'(t) ] dt=0$, and $E[I^*(X)]=\frac{\pi}{1+\rho}$.
\item If $0 \leqslant \pi \leqslant \widehat{\pi}$, then
\begin{align*}
   {I}^*(z)=\begin{cases}
            0, &\quad  \text{  if  }            0\leqslant z < F^{-1}_X(f),\\
            z-F^{-1}_X(f),&\quad  \text{  if  }             F^{-1}_X(f)\leqslant z \leqslant M,
          \end{cases}
\end{align*}
where $f$ is the unique scalar satisfying $E[I^*(X)]=\frac{\pi}{1+\rho}$.
\end{enumerate}
\end{thm}
\begin{proof}
The result follows from Propositions \ref{prop:51}, \ref{prop:52} and the fact that $K_\Delta$ is a constant for any $0<\Delta<E[X]$.
\end{proof}
%\begin{remark}
%Let us compare Theorem 3.10 and Theorem 3.6 under the same probability weighting function. It is easy to see that $\pi_c<\widehat{\pi}$ since $l_\Delta<c$. So, for $\pi\in (\pi_c,\widehat{\pi})$, an insured with the RDEU type preference likes to get the deductible contract while that of the Yaari's type wants to get contract \eqref{three}. The reason is that an insured with a strictly concave utility function is more risk averse; hence he is more concerned about large losses.
%\end{remark}
\par
Now, we study the case in which $A_u(z)$ is strictly decreasing. We need the following lemma.
\begin{lemma}\label{lemma:55}
If $0<\Delta_1<\Delta_2<E[X]$, then $a<l_{\Delta_1}<l_{\Delta_2}<c$.
\end{lemma}
\begin{proof}
According to definition of $l_{\Delta_1}$, we have
\begin{align*}
h_{\Delta_1}(l_{\Delta_1})=\int_0^{l_{\Delta_1}} \left[ \frac{u'(W_{\Delta_1}-F^{-1}_X(l_{\Delta_1}))(1-T(l_{\Delta_1}))}{1-l_{\Delta_1}}-u'(W_{\Delta_1}-F^{-1}_X(t))T'(t) \right] dt=0.
\end{align*}
Since $W_{\Delta_1}<W_{\Delta_2}$, we have
$\frac{u'(W_{\Delta_2}-F^{-1}_X(l_{\Delta_1}))}{u'(W_{\Delta_1}-F^{-1}_X(l_{\Delta_1}))}<\frac{u'(W_{\Delta_2}-F^{-1}_X(t))}{u'(W_{\Delta_1}-F^{-1}_X(t))}$ for $t\in [0,l_{\Delta_1})$ by Lemma \ref{lemma:A2} in Appendix A. Hence
\begin{align*}
h_{\Delta_2}(l_{\Delta_1})=\int_0^{l_{\Delta_1}} \left[ \frac{u'(W_{\Delta_2}-F^{-1}_X(l_{\Delta_1}))(1-T(l_{\Delta_1}))}{1-l_{\Delta_1}}-u'(W_{\Delta_2}-F^{-1}_X(t))T'(t) \right] dt<0.
\end{align*}
As a result $h_{\Delta_2}(l_{\Delta_1})<0$, $h_{\Delta_2}(c)>0$. Since $h_{\Delta_2}'(z)>0$ for $z\in [l_{\Delta_1},c)$, we get $l_{\Delta_2}\in (l_{\Delta_1},c)$.
\end{proof}
\par
Define $\Delta(d):=\int_0^d F^{-1}_X(z) dz+\int_d^1 F^{-1}_X(d) dz=\int_0^d F^{-1}_X(z) dz+F^{-1}_X(d)(1-d)$ on $d\in [a,c]$. Then $\Delta'(d)=(1-d)(F^{-1}_X)'(d)>0$. Hence, $\Delta(\cdot)$ is a continuous and strictly increasing function. Determine $l_{\Delta(a)}$ and $l_{\Delta(c)}$ by $h_{\Delta(a)}(l_{\Delta(a)})=0$ and $h_{\Delta(c)}(l_{\Delta(c)})=0$, and  set $\widetilde{\Delta}:=\Delta(l_{\Delta(a)})$ and $\overline{\Delta}:=\Delta(l_{\Delta(c)})$. Finally, define a function $g(\cdot)$ on $[a,c]$ as follows:
\begin{align*}
g(z):=\int_0^z \left[ \frac{u'(W_0+(1+\rho)\Delta(z)-F^{-1}_X(z))(1-T(z))}{1-z}-u'(W_0+(1+\rho)\Delta(z)-F^{-1}_X(t))T'(t) \right] dt.
\end{align*}

\begin{prop}\label{prop:53}
Assume that Assumptions \ref{assump:21}, \ref{assump:23}, and \ref{ap} hold and that $A_u(\cdot)$ is strictly decreasing. Then the optimal solution  to \eqref{orgi-3}
is given as
\begin{enumerate}[(i)]
\item
 If $\Delta = 0$, then ${G}^*(z)=0$ $\forall 0\leqslant z\leqslant 1$.
\item If $0<\Delta \leqslant \widetilde{\Delta}$, then
\begin{align}
   G^*(z)=\begin{cases}
            0, &\quad  \text{  if  }              0 \leqslant z< z_2,\\
            F^{-1}_X(z)-F^{-1}_X(z_2), &\quad  \text{  if  }              z_2 \leqslant z< z_1,\\
            F^{-1}_X(z_1)-F^{-1}_X(z_2),&\quad  \text{  if  }            z_1 \leqslant z \leqslant 1,
          \end{cases}
\end{align}
where $z_2,z_1$ are such that $z_2\leqslant a \leqslant z_1$, $\int_{z_2}^{z_1} [ \frac{u'(W_\Delta-F^{-1}_X(z_1)+F^{-1}_X(z_2))(1-T(z_1))}{1-z_1}- u'(W_\Delta-F^{-1}_X(t)+F^{-1}_X(z_2))T'(t) ] dt=0$, and $\int_0^1 G^*(z) dz=\Delta$.
\item If $\widetilde{\Delta}< \Delta <\overline{\Delta}$, then let $p\in (l_{\Delta(a)},l_{\Delta(c)})$ such that $\Delta(p)=\Delta$. If $g(p)<0$, then
\begin{align}
   G^*(z)=\begin{cases}
            0, &\quad  \text{  if  }              0 \leqslant z< z_2,\\
            F^{-1}_X(z)-F^{-1}_X(z_2), &\quad  \text{  if  }              z_2 \leqslant z< z_1,\\
            F^{-1}_X(z_1)-F^{-1}_X(z_2),&\quad  \text{  if  }            z_1\leqslant  z \leqslant 1,
          \end{cases}
\end{align}
where $z_2,z_1$ are such that $z_2\leqslant a \leqslant z_1$, $\int_{z_2}^{z_1} [ \frac{u'(W_\Delta-F^{-1}_X(z_1)+F^{-1}_X(z_2))(1-T(z_1))}{1-z_1}- u'(W_\Delta-F^{-1}_X(t)+F^{-1}_X(z_2))T'(t) ] dt=0$, and $\int_0^1 G^*(z) dz=\Delta$. If $g(p)\geqslant 0$, then
\begin{align}
   G^*(z)=\begin{cases}
            F^{-1}_X(z), &\quad  \text{  if  }           0\leqslant z < f,\\%[3pt]
             F^{-1}_X(f),&\quad  \text{  if  }            f\leqslant z \leqslant 1,
          \end{cases}
\end{align}
where $f$ is such that $f < l_{\Delta(c)}$ and $\int_0^1 G^*(z) dz = \Delta$.
\item If $\overline{\Delta} \leqslant \Delta \leqslant E[X] $, then
\begin{align}
   G^*(z)=\begin{cases}
            F^{-1}_X(z), &\quad  \text{  if  }           0\leqslant z < f,\\ %[3pt]
             F^{-1}_X(f),&\quad  \text{  if  }            f\leqslant z \leqslant 1,
          \end{cases}
\end{align}
where $f$ is such that $f \geqslant l_{\Delta(c)}$ and $\int_0^1 G^*(z) dz = \Delta$.
\end{enumerate}
\end{prop}
\begin{proof}
(i),(ii) and (iv) are direct consequences of Proposition \ref{prop:51}, \ref{prop:52} and Lemma \ref{lemma:55}. For
(iii), there is a unique $p\in (l_{\Delta(a)},l_{\Delta(c)})$ such that $\Delta(p)=\Delta$, which follows from the definition of $\widetilde{\Delta}$, $\overline{\Delta}$ and the fact that $\Delta(\cdot)$ is a continuous and strictly increasing function. If $g(p)<0$, then $h_\Delta(p)<0$; hence $l_\Delta>p$. Therefore, $\Delta<K_\Delta$. The desired result follows from Proposition \ref{prop:52}. The proof for $g(p)\geqslant 0$ is similar.
\end{proof}
\par
Let us give the result in terms of premium and indemnity function.
\begin{thm}
Assume that Assumptions \ref{assump:21}, \ref{assump:23}, and \ref{ap} hold, and that $A_u(\cdot)$ is strictly decreasing. Then the optimal indemnity function $I^*(\cdot)$ to Problem (1)
is given as
\begin{enumerate}[(i)]
\item If $ \pi=(1+\rho)E[X]$, then ${I}^*(z)=z$ $\forall z\in [0,M]$.
\item If $(1+\rho)(E[X]-\widetilde{\Delta})\leqslant \pi <(1+\rho)E[X]$, then
\begin{align*}
   {I}^*(z)=\begin{cases}
            z, &\quad  \text{  if  }           0\leqslant z <F^{-1}_X(z_2),\\ %[3pt]
             F^{-1}_X(z_2) , &\quad  \text{  if  }        F^{-1}_X(z_2)\leqslant z < F^{-1}_X(z_1),\\
            z-F^{-1}_X(z_1)+F^{-1}_X(z_2),&\quad  \text{  if  }             F^{-1}_X(z_1)\leqslant z \leqslant M,
          \end{cases}
\end{align*}
where $(z_2,z_1)$ is the unique pair satisfying $z_2\leqslant a \leqslant z_1$,
$\int_{z_2}^{z_1} [ \frac{u'(W_\Delta-F^{-1}_X(z_1)+F^{-1}_X(z_2))(1-T(z_1))}{1-z_1}- u'(W_\Delta-F^{-1}_X(t)+F^{-1}_X(z_2))T'(t) ] dt=0$, and $E[I^*(X)]=\frac{\pi}{1+\rho}$.
\item  If $ (1+\rho)(E[X]-\overline{\Delta})< \pi <(1+\rho)(E[X]-\widetilde{\Delta})$, then let $p\in (l_{\Delta(a)},l_{\Delta(c)})$ such that $\Delta(p)=E[X]-\frac{\pi}{1+\rho}$. If $g(p)<0$, then
\begin{align*}
   {I}^*(z)=\begin{cases}
            z, &\quad  \text{  if  }           0\leqslant z < F^{-1}_X(z_2),\\%[3pt]
             F^{-1}_X(z_2) , &\quad  \text{  if  }        F^{-1}_X(z_2)\leqslant z < F^{-1}_X(z_1),\\
            z-F^{-1}_X(z_1)+F^{-1}_X(z_2),&\quad  \text{  if  }             F^{-1}_X(z_1)\leqslant z \leqslant M,
          \end{cases}
\end{align*}
where $(z_2,z_1)$ is the unique pair satisfying $z_2\leqslant a \leqslant z_1$,
$\int_{z_2}^{z_1} [ \frac{u'(W_\Delta-F^{-1}_X(z_1)+F^{-1}_X(z_2))(1-T(z_1))}{1-z_1}- u'(W_\Delta-F^{-1}_X(t)+F^{-1}_X(z_2))T'(t) ] dt=0$, and $E[I^*(X)]=\frac{\pi}{1+\rho}$. If $g(p)\geqslant 0$, then
\begin{align*}
   {I}^*(z)=\begin{cases}
            0, &\quad  \text{  if  }            0\leqslant z < F^{-1}_X(f),\\
            z-F^{-1}_X(f),&\quad  \text{  if  }             F^{-1}_X(f)\leqslant z \leqslant M,
          \end{cases}
\end{align*}
where $q$ is the unique number satisfying $f<l_{\Delta(c)}$ and $E[I^*(X)]=\frac{\pi}{1+\rho}$.
\item If $ 0\leqslant \pi \leqslant (1+\rho)(E[X]-\overline{\Delta})$, then
\begin{align*}
   {I}^*(z)=\begin{cases}
            0, &\quad  \text{  if  }            0\leqslant z < F^{-1}_X(f),\\
            z-F^{-1}_X(f),&\quad  \text{  if  }             F^{-1}_X(f)\leqslant z \leqslant M,
          \end{cases}
\end{align*}
where $q$ is the unique number satisfying $f\geqslant l_{\Delta(c)}$ and $E[I^*(X)]=\frac{\pi}{1+\rho}$.
\end{enumerate}
\end{thm}
\begin{proof}
It follows easily from Proposition \ref{prop:53}.
\end{proof}

\section{Numerical Illustrations}
\noindent
In this section, we use a numerical example to illustrate our result with a given premium $\pi$. we take the same numerical setting as in Bernard et al. (2015) for a comparison purpose:
The loss $X$ follows a truncated exponential distribution with the density function $f(x)=\frac{me^{-mx}}{1-e^{-mM}}$, where the intensity parameter $m=0.1$, and $M=10$. The initial wealth $W_0=15$, and $u(x)=1-e^{-\gamma x}$ with $\gamma=0.02$. Moreover, $\rho=0.2$ and $\pi=3$.
Finally,
the weighting function $T_\theta(x)=\frac{x^\theta}{(x^\theta+(1-x)^\theta)^{\frac{1}{\theta}}}$ with $\theta=0.5$.
We can verify that the assumptions  of Theorem \ref{thm:54}  is satisfied under this setting.
\begin{figure}[htp]
  \centering
      \includegraphics[width=4 in]{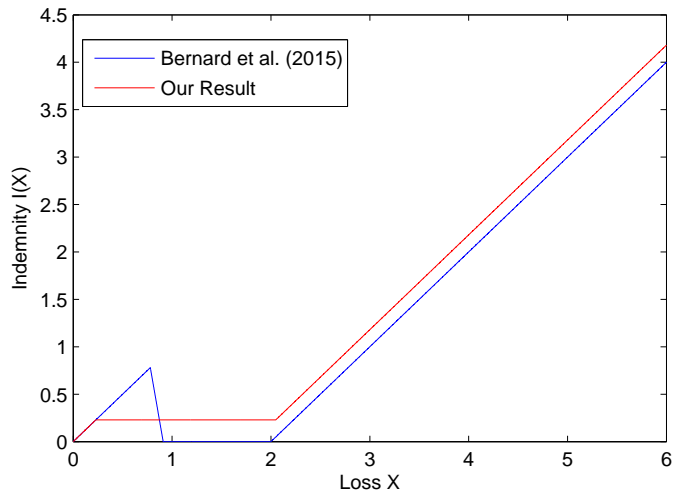}
      \caption{Our result vs Bernard et al. (2015). }
\end{figure}
\par
The optimal indemnity obtained by Bernard et al. (2015) is plotted in blue in Figure 3. We note that  in the result of Bernard et al. (2015),
if the loss is between 1 and 2, then the insured has the incentive to hide part of the loss in order to be paid with a larger compensation.
By contrast, our indemnity function, depicted in red, is increasing and any increment in compensations is always less than or equal to the increment in losses. It effectively rules out the aforementioned behavior of moral hazard.

\section{Conclusion}
\noindent
In this paper, we have studied an optimal insurance design problem where the insured has the RDU preference. There are documented evidences proving that this preference captures human behaviors better than the EU preference. The main contribution of our work is that our optimal contracts
are monotone with respect to losses, thereby eliminating  the potential problem of moral hazard associated with the existing results.
\par
An interesting conclusion from our results is that, under our assumptions (in particular Assumption \ref{ap}-(ii)), there are only two types of non-trivial optimal contracts possible, one being
the classical deductible and the other the threefold contract covering both small and large losses.
\par
While we have demonstrated that Assumption \ref{ap}-(ii) holds for many economically interesting cases, removing this assumption remains a mathematically outstanding open problem.

%In our work, the indemnity function is required to be non-decreasing and any increment in compensation is always less than or equal to the increment in loss. We represent this requirement by using quantile function, and then apply the calculus of variations to solve this problem. The Yaari model and general utility model are studied respectively in this paper. After solving this problem, we find a class of contracts, which have a qualitatively different feature from the classic insurance contracts and BHYZ's contracts. The practitioners may benefit from considering our results.
%This quantile formulation used in our work is very general and effective way in dealing this monotonicity constraint.
%%%%%%%%%%%%%%%%%%%%%%%%%%%%

\begin{appendix}
\section{Some Lemmas}
\noindent
In this part, we prove some lemmas which have been used in Section 5.
\begin{lemma}\label{lemma:A1}
Assume $T(\cdot):[0,1]\mapsto [0,1]$ satisfies Assumption \ref{assump:23}. We have the following results: \\
(i)  If $a<z$, then $T'(z)<\frac{1-T(z)}{1-z}$. \\
(ii) If $a<z_2<b$ and $z_2<z_1<1$, then $\frac{1-T(z_1)}{1-z_1}> \frac{T(z_1)-T(z_2)}{z_1-z_2}$.
\end{lemma}
\begin{proof}
(i) If $a<z\leqslant b$, then $T'(z)<T'(a)<\frac{1-T(z)}{1-z}$. If $b<z$, then $T'(z)<\frac{1-T(z)}{1-z}$ since $T(\cdot)$ is convex and strictly increasing on $[b,1]$. \\
(ii) Since $1-T(z_1)$, $1-z_1$, $T(z_1)-T(z_2)$, and $z_1-z_2$ are all strictly positive, we have
\begin{eqnarray*}
 \frac{1-T(z_1)}{1-z_1}> \frac{T(z_1)-T(z_2)}{z_1-z_2}
\Longleftrightarrow && \frac{1-T(z_1)}{1-z_1}> \frac{(1-T(z_1))+(T(z_1)-T(z_2))}{(1-z_1)+(z_1-z_2)} \\
\Longleftrightarrow && \frac{1-T(z_1)}{1-z_1}> \frac{1-T(z_2)}{1-z_2}.
\end{eqnarray*}
However, $\frac{1-T(z_1)}{1-z_1}> \frac{1-T(z_2)}{1-z_2}$ follows from Lemma \ref{desp}.
\end{proof}
\par
For fixed $x>0$, define $q(z):=u'(x+z)u'(x-z)$ on $z\in(0,x)$.
\begin{lemma}\label{lemma:A2}
If $-\frac{u''(z)}{u'(z)}$ is strictly decreasing, then $q(z)$ is a strictly increasing function on $z\in (0,x)$.
\end{lemma}
\begin{proof}
We take derivative:
\begin{eqnarray*}
q'(z)= && u''(x+z)u'(x-z)-u'(x+z)u''(x-z)  \\
= && u'(x+z)u'(x-z)\left[(-\frac{u''(x-z)}{u'(x-z)})-(-\frac{u''(x+z)}{u'(x+z)})\right]>  0.
\end{eqnarray*}
Hence, we get the result.
\end{proof}

\section{Existence of Optimal Solutions to \eqref{orgi-3} and \eqref{orgi-6} }
\noindent
%For the following lemma, refer to Dunford and Schwartz (1958, page 382).
%\begin{lemma}
%(Arzela-Ascoli Theorem) \ Consider a sequence of real-valued continuous functions $(f_n)_{n \in \mathds{N}}$ defined on a closed and bounded interval $[a,b]$ of the real line. If this sequence is uniformly bounded and equicontinuous, then there exists a subsequence $f_{n_k}$ that converges uniformly to a continuous function on $[a,b]$.
%\end{lemma}
We first prove that the constraint set $\mathbb{G}$ is compact under some norm. We consider all the continuous functions on $[0,1]$, denoted as $C[0,1]$. Define a metric between $x(\cdot),y(\cdot)$ as $\rho(x(\cdot),y(\cdot))=\max_{0\leqslant t \leqslant 1}|x(t)-y(t)|$, $\forall x(\cdot),y(\cdot) \in C[0,1]$. Clearly, $C[0,1]$ is a metric space under $\rho$. By Arzela--Ascoli's theorem, for any sequence $(G_n(\cdot))_{n \in \mathds{N}}$ in ${\mathbb{G}}$, there exists a subsequence $G_{n_k}(\cdot)$ that converges in $C[0,1]$ under $\rho$.
\begin{lemma}\label{lemma:B1}
The feasible set $\ {\mathbb{G}}$ is compact  under $\rho$.
\end{lemma}
\begin{proof}
For any sequence $(G_n(\cdot))_{n \in \mathds{N}}$ in ${\mathbb{G}}$ , there exists a subsequence $G_{n_k}(\cdot)$ that uniformly converges in $G^*(\cdot) \in C[0,1]$. We now prove that $G^*(\cdot) \in {\mathbb{G}}$. If there exist $a>b$ such that $G^*(b)-G^*(a)=\eta > 0$, then take $\varepsilon :=\frac{1}{3} \eta$. If follows from the uniform convergence that  there exists $K$ such that $\rho (G_{n_{k}}(\cdot),G^*(\cdot)) \leqslant \varepsilon$ $\forall \ k \geqslant K$. Hence, $0 < \eta = G^*(b)-G^*(a)=G^*(b)-G_{n_{k}}(b)+G_{n_{k}}(b)-G_{n_{k}}(a)+G_{n_{k}}(a)-G^*(a) \leqslant \varepsilon + 0 + \varepsilon=\frac{2}{3}\eta$ $\forall \ k \geqslant K$, which is a contradiction. This proves  that $0 \leqslant G^*(a)-G^*(b)$ $\forall a>b$. Similarly, we can prove that $G^*(a)-G^*(b) \leqslant F^{-1}_X(a)-F^{-1}_X(b)$.
\end{proof}
\par
The existence of optimal solutions to \eqref{orgi-3} and \eqref{orgi-6} can be established now. For example, for \eqref{orgi-6}, let $v_\lambda(\Delta)$ be the optimal value of \eqref{orgi-6} under given $\lambda$ and $\Delta$. We can take a sequence $(G_n(\cdot))_{n \in \mathds{N}}$ in ${\mathbb{G}}$ such that $v_\lambda(\Delta) = \lim_{n \uparrow +\infty} U_\Delta(\lambda, G_n(\cdot))$. Then, according to Lemma \ref{lemma:B1}, there exists a subsequence $G_{n_k}(\cdot)$ converging to $G^*(\cdot)$ in ${\mathbb{G}}$ and $G^*(\cdot)$ is optimal solution to \eqref{orgi-6}. For \eqref{orgi-3}, the proof is similar.

\section{Existence of Lagrangian Multiplier to \eqref{orgi-3}}
\noindent
For the following lemma, refer to Komiya (1988) for an elementary proof.
\begin{lemma}
(Sion's Minimax Theorem) \ Let X be a compact convex subset of a linear topological space and Y a convex subset of a linear topological space. If $f$ is a real-valued function on $X\times Y$ such that $f(x,\cdot )$ is continuous and concave on $Y$ $\forall x\in X$, and $f(\cdot ,y)$ is continuous and convex on $X$ $\forall y\in Y$, then, $\min\limits_{{x\in X}}  \  \max\limits_{{y\in Y}} f(x,y)=\max\limits_{{y\in Y}} \ \min\limits_{{x\in X}}f(x,y).$
\end{lemma}
\begin{prop}\label{prop-dual}
For any $0<\Delta<E[X]$, there is $\lambda^*$ such that $\widetilde{G}_{\lambda^*}(\cdot)$ is optimal solution to \eqref{orgi-6} under $\lambda^*$ and $\int_0^1 \widetilde{G}_{\lambda^*}(z) dz = \Delta$.
\end{prop}
\begin{proof}
Let $\Delta$ be given with $0<\Delta<E[X]$. Denote by $G^*(\cdot)$ the optimal solution to \eqref{orgi-3} under $\Delta$ (it is easy to show $\int_0^1 G^*(z) dz=\Delta$) and by $\widetilde{G}_{\lambda}(\cdot)$ the optimal solution to \eqref{orgi-6} under $\lambda$ and $\Delta$. Denote by $v(\Delta)$ and $v(\lambda,\Delta)$ be respectively the optimal values of \eqref{orgi-3} and \eqref{orgi-6}.
\par
We first prove that $v(\lambda,\Delta)$ is a convex function in $\lambda$ for given $\Delta$. Noting that $U_\Delta(\lambda,G(\cdot))$ is linear in $\lambda$ for any given $G(\cdot)$, we have
\begin{eqnarray*}
v(\alpha\lambda_1+(1-\alpha)\lambda_2,\Delta) = && \max\limits_{ {G(\cdot)}} U_\Delta(\alpha\lambda_1+(1-\alpha)\lambda_2,G(\cdot)) \\
= && \max\limits_{ {G(\cdot)}} \{\alpha U_\Delta(\lambda_1,G(\cdot))+(1-\alpha)U_\Delta(\lambda_2,G(\cdot))\} \\
\leqslant && \max\limits_{ {G(\cdot)}} \{\alpha U_\Delta(\lambda_1,G(\cdot))\}+\max\limits_{ {G(\cdot)}} \{(1-\alpha) U_\Delta(\lambda_2,G(\cdot))\} \\
= && \alpha \max\limits_{ {G(\cdot)}} \{U_\Delta(\lambda_1,G(\cdot))\}+(1-\alpha)\max\limits_{ {G(\cdot)}} \{U_\Delta(\lambda_2,G(\cdot))\} \\
= && \alpha v(\lambda_1,\Delta)+(1-\alpha)v(\lambda_2,\Delta).
\end{eqnarray*}
\par
Moreover, %we have $\min\limits_{0\leqslant \lambda} \max\limits_{ {G(\cdot)\in \mathbb{G}}} U_\Delta(\lambda,G(\cdot))\geqslant U_\Delta(G^*(\cdot))\geqslant \max\limits_{ {G(\cdot)\in \mathbb{G}}} \min\limits_{0\leqslant \lambda} U_\Delta(\lambda,G(\cdot)).$ Then,
by Sion's minimax theorem, the following equality holds: $\max\limits_{0\leqslant \lambda} \min\limits_{ {G(\cdot)\in \mathbb{G}}} -U_\Delta(\lambda,G(\cdot))=\min\limits_{ {G(\cdot)\in \mathbb{G}}} \max\limits_{0\leqslant \lambda} -U_\Delta(\lambda,G(\cdot))$; hence $\min\limits_{0\leqslant \lambda} \max\limits_{ {G(\cdot)\in \mathbb{G}}} U_\Delta(\lambda,G(\cdot))=\max\limits_{ {G(\cdot)\in \mathbb{G}}} \min\limits_{0\leqslant \lambda} U_\Delta(\lambda,G(\cdot))$. Finally, we have $v(\Delta)= \inf_{0\leqslant \lambda} v(\lambda,\Delta)$ (i.e. $\min\limits_{0\leqslant \lambda} \max\limits_{ {G(\cdot)\in \mathbb{G}}} U_\Delta(\lambda,G(\cdot))= U_\Delta(G^*(\cdot))$).
\par
Let us denote $\overline{\lambda}:=\frac{v(\Delta)+1}{\int_0^1 F^{-1}_X(z) dz - \Delta}=\frac{U_\Delta(G^*(\cdot))+1}{E[X] - \Delta}$. For any $\lambda\geqslant \overline{\lambda}$, we have
\begin{eqnarray*}
v(\lambda,\Delta)= && \max\limits_{ {G(\cdot)\in \mathbb{G}}} U_\Delta(\lambda,G(\cdot))
\geqslant  U_\Delta(\lambda,F^{-1}_X(z))) \\
= && \int_0^1 u(W_\Delta-F^{-1}_X(z))T'(z)dz + \lambda(\int_0^1 F^{-1}_X(z) dz - \Delta)
\geqslant\lambda(\int_0^1 F^{-1}_X(z) dz - \Delta) \\
\geqslant && \overline{\lambda}(\int_0^1 F^{-1}_X(z) dz - \Delta) \ \ (\mbox{since }\int_0^1 F^{-1}_X(z) dz > \Delta)\\
= && v(\Delta)+1,
\end{eqnarray*}
which yields $v(\Delta)= \inf_{0\leqslant \lambda } v(\lambda,\Delta)= \inf_{0\leqslant \lambda\leqslant \overline{\lambda}} v(\lambda,\Delta)$.
%\footnote{{\bf Zhou: This is a bit strange: why do you use $v(\Delta)+1$? Shouldn't $v(\Delta)$ be more normal? Zhuang: This is because we need a low bound to ensure that $v(\Delta)= \inf_{0\leqslant \lambda\leqslant \overline{\lambda}} v(\lambda,\Delta)$ holds. Here, 1 can be changed into a strictly positive number $\epsilon$. If we only have $v(\lambda,\Delta)\geq v(\Delta)$, then $\lambda^*$ which minimizes $v(\lambda,\Delta)$ still can be located above $\bar{\lambda}$. }}
\par
Therefore, by using the convexity of $v(\lambda,\Delta)$, we can find the optimal $\lambda^*\in [0,\overline{\lambda}]$ minimizes the right part, and satisfies that $v(\Delta)=v(\lambda^*,\Delta)$. Moreover,
\begin{eqnarray*}
v(\lambda^*,\Delta)\geqslant && U_\Delta(\lambda^*,G^*(\cdot))
=  \int_0^1 u(W_\Delta-G^*(z))T'(z)dz+\lambda^*(\int_0^1 G^*(z) dz - \Delta) \\
= && \int_0^1 [u(W_\Delta-G^*(z))T'(z)] dz
= U_\Delta(G^*(\cdot))=v(\Delta).
\end{eqnarray*}
The second equality comes from the fact that $G^*(\cdot)$ is the optimal solution to \eqref{orgi-3} under $\Delta$; hence $\int_0^1 G^*(z) dz=\Delta$. By $v(\Delta)=v(\lambda^*,\Delta)$ and $v(\lambda^*,\Delta)\geqslant U_\Delta(\lambda^*,G^*(\cdot))=v(\Delta)$, we have $G^*(\cdot)$ is optimal solution to \eqref{orgi-6} under given $\lambda^*$. And, by uniqueness of optimal solutions to \eqref{orgi-6}, we know that $G^*(\cdot)$ is the unique optimal solution to \eqref{orgi-6} under given $\lambda^*$ and satisfying $\int_0^1 G^*(z) dz=\Delta$.
\end{proof}
\end{appendix}

\newpage

%%%%%%%%%%%%%%%%%%%%%%%%%%%%%%%%%%%%%%%%%%%%%%%%%%

%%%%%%%%%%%%%%%%%%%%%%%%%%%%%%%%%%%%%%%%%%%%%%%%%

%\nocite{*}
%\bibliographystyle{abbrv}
%\bibliography{reference}

\end{document}